%% file: eids-pointmutation.tex
\documentclass[smallextended]{svjour3}       
\pdfoutput=1
\smartqed  
\usepackage[authoryear]{natbib}
\usepackage{amssymb,amsmath,pdfpages}
\usepackage[all]{xy}

\journalname{Journal of Mathematical Biology}

\newcommand{\bN}{\mathbb{N}}

\newcommand{\bR}{\mathbb{R}}

\newcommand{\bE}{\mathbb{E}}

\newcommand{\cH}{\mathcal{H}}

\newtheorem{hypothesis}{Hypothesis}

\newcommand{\grant}{EPSRC grant EP/H031936/1 and BBSRC grant BB/L009579/1}

\begin{document}

\title{Monotonicity of Fitness Landscapes and Mutation Rate Control\thanks{This work was supported by \grant}}

\author{Roman V. Belavkin \and
Alastair Channon \and
Elizabeth Aston \and
John Aston \and
Rok Kra{\v s}ovec \and
Christopher G. Knight}

\institute{Roman V. Belavkin
\at School of Science and Technology, Middlesex University, London NW4 4BT, UK\\\email{r.belavkin@mdx.ac.uk}
\and
Alastair Channon \and Elizabeth Aston
\at School of Computing and Mathematics, Keele University, ST5 5BG, UK\\ \email{a.d.channon@keele.ac.uk, e.j.aston@keele.ac.uk}
\and
John A. D. Aston
\at Statslab, University of Cambridge, CB3 0WB, UK\\\email{j.aston@statslab.cam.ac.uk}
\and
Rok Kra{\v s}ovec \and Christopher G. Knight
\at Faculty of Life Sciences, University of Manchester, M13 9PT, UK\\\email{rok.krasovec@manchester.ac.uk, chris.knight@manchester.ac.uk}}

\date{Received: 18 December 2014 / Accepted: 12 April 2016}

\maketitle



\begin{abstract}
A common view in evolutionary biology is that mutation rates are minimised. However, studies in combinatorial optimisation and search have shown a clear advantage of using variable mutation rates as a control parameter to optimise the performance of evolutionary algorithms.  Much biological theory in this area is based on Ronald Fisher's work, who used Euclidean geometry to study the relation between mutation size and expected fitness of the offspring in infinite phenotypic spaces.  Here we reconsider this theory based on the alternative geometry of discrete and finite spaces of DNA sequences.  First, we consider the geometric case of fitness being isomorphic to distance from an optimum, and show how problems of optimal mutation rate control can be solved exactly or approximately depending on additional constraints of the problem.  Then we consider the general case of fitness communicating only partial information about the distance.  We define weak monotonicity of fitness landscapes and prove that this property holds in all landscapes that are continuous and open at the optimum.  This theoretical result motivates our hypothesis that optimal mutation rate functions in such landscapes will increase when fitness decreases in some neighbourhood of an optimum, resembling the control functions derived in the geometric case.  We test this hypothesis experimentally by analysing approximately optimal mutation rate control functions in 115 complete landscapes of binding scores between DNA sequences and transcription factors.  Our findings support the hypothesis and find that the increase of mutation rate is more rapid in landscapes that are less monotonic (more rugged). We discuss the relevance of these findings to living organisms.

\end{abstract}


\section{Introduction}
\label{sec:intro}

Mutation is one of the most important biological processes that influence evolutionary dynamics.  During replication mutation leads to a loss of information between the offspring and its parent, but it also allows the offspring to acquire new features. These features are likely to be deleterious, but have the potential to be beneficial for adaptation.  Thus, mutation can be seen as a process of innovation, which is particularly important as the number of all living organisms is tiny relative to the number of all possible organisms.  A question that naturally arises with regards to mutation is whether there is an optimal balance between the amount of information lost and potential fitness gained.

The seminal mathematical work to investigate biological mutation is by Ronald \citet{Fisher30}, who considered mutation as a random motion in Euclidean space, the points of which are vectors representing collections of phenotypic traits of organisms.  Using the geometry of Euclidean space, Fisher showed that probability of adaptation decreases exponentially as a function of mutation size (defined using the ratio of mutation radius and distance to the optimum), and concluded, therefore, that adaptation is more likely to occur by small mutations.  Several studies, however, suggested that large mutations can be quite frequent in nature, thereby prompting re-examination of the theory \citep{Orr05}.  Thus, \citet{Kimura80} extended the theory to take into account differences in probabilities of fixation for mutations of small and large size.  Subsequently \citet{Orr98} considered the effect of mutation across several replications.  Interestingly, while Fisher had a critical role in developing mathematical theory around discrete alleles, in his geometric model he used Euclidean space of traits as the domain of mutation, which is uncountably infinite and unbounded.  This important issue only became apparent after the realisation that biological evolution occurs in a countable or even finite space of discrete molecular sequences \citep{Smith70}.  However, subsequent geometric models based on Fisher's, while explicitly modelling discrete mutational steps \citep[e.g.][]{Orr02}, continue to assume that they occur within the same infinite Euclidean space.  This issue may contribute to the fact that the predictions of such models have at best only been partially verified in actual biological systems \citep{McDonald11,Bataillon11,Kassen06,Rokyta08}.  In this and previous work, we consider mutation using the geometry and combinatorics of Hamming spaces \citep{Belavkin_etal11:_ecal11,Belavkin11:_itw11}, which are finite, and this leads to a radically different view about the role of large mutations.


Independent of such biological concerns, researchers in evolutionary computation and operations research have a long history of considering variable mutation rates in genetic algorithms (GAs) \citep[e.g. see][for reviews]{Eiben_etal99,Ochoa02,Falco_etal02,Cervantes-Stephens06,Vafaee_etal10}.  In particular, \citet{Ackley87} suggested that mutation probability is analogous to temperature in simulated annealing, which decreases with time through optimisation.  A gradual reduction of mutation rate was also proposed by \citet{Fogarty89}.   Markov chain analysis of GAs was used by \citet{Yanagiya93} to show that a sequence of optimal mutation rates maximising the probability of obtaining the global solution exists in any problem.  In particular, \citet{Back93} studied the probability of adaptation in the space of binary strings and derived optimal mutation rates depending on the distance from the global optimum.  More recently, numerical methods have been used to optimise a mutation operator \citep{Vafaee_etal10} that was based on the Markov chain model of GA by \citet{Nix-Vose92}, although the complexity of this method may restrict its application to small spaces and populations.  More recently, several authors have analysed the run-time of the co-called $(1+1)$-evolutionary algorithm using constant and adaptive mutation rates and demonstrating some advantages of the latter \citep{Suntje_etal2010,Sutton_etal2011}.  Thus, the idea of using variable mutation rates to optimise evolutionary dynamics is not new.  Unfortunately, these results in the field of evolutionary computation (EC) have specific computational focus, which limits their appeal for biology.

First, theoretical work on EC has focused almost exclusively on systems of binary strings.  Optimisation of mutation rates of DNA strings, which have the alphabet of four bases, involves analysis of a significantly more difficult combinatorics and geometry.  Previously, we presented some results on optimal mutation rates \citep{Belavkin_etal11:_ecal11,Belavkin11:_itw11}, which used formula~(\ref{eq:h-intersection}) for the intersection of spheres in general Hamming spaces.  Here we give the derivation of this formula in Appendix~\ref{sec:spheres} and show how it can be used to generalise Fisher's geometric model of adaptation in Section~\ref{sec:geometry}.

Second, the run-time analysis and optimisation of evolutionary algorithms is concerned with their long term behaviour, which may have little relevance for biological systems.  For example, \citet{Suntje_etal2010} show that the run time of the $(1+1)$-evolutionary algorithm is on the order of $l^2$, where $l$ is the length of a binary string.  In biological organisms, the typical length of DNA sequence is $l\in[10^8,10^{11}]$ (and the alphabet size is $\alpha=4$).  Assuming the minimum of 20 minutes between replications, the run-time of order $l^2$ will significantly exceed $10^{14}$ years --- the estimated time after which stars will cease to exist in the Universe \citep{Adams-Laughlin97}.  Moreover, biological landscapes may fail to have a global optimum to converge to, because the set of all DNA sequences with variable lengths is infinite.  In addition, biological landscapes are not static, and change on a regular basis.  Thus, the short-term behaviour, perhaps within one or several replications, is more important for optimisation of parameters in biological systems.  Here we develop these insights regarding mutation rate variation towards the particular issues presented by biological systems.

In Section~\ref{sec:geometry} we show how the problem of optimal control of mutation rate can be defined in different ways leading to different solutions.  In some cases, these solutions can be obtained analytically.  For example, in the idealised geometric model, when maximisation of fitness is equivalent to minimisation of distance to a global optimum, the optimal mutation rates can be derived as functions of the distance \citep{Belavkin11:_dyninf,Belavkin11:_qbic11}.  This, however, is not the case for more realistic landscapes, which can be rugged.  In Section~\ref{sec:meta-ga}, we address how the control functions can be obtained numerically.  Although fitness landscapes have been analysed and classified in terms of hardness for evolutionary algorithms \citep{He_etal2015}, there is no general theory about optimal mutation rates in arbitrary landscapes.  The development of such theory is the main focus and contribution of this paper.  In Section~\ref{sec:monotonic}, we consider a fitness landscape as a communication channel between fitness values and distances from a nearest optimum.  We introduce various notions of monotonicity of a fitness landscape, and discuss how these properties are related to the genotype-phenotype mapping.  The main theoretical result is a theorem about weak monotonicity of continuous landscapes, which establishes the condition for a similarity between fitness and distance to an optimum in a broad class of landscapes.  This suggests a similarity between fitness-based and distance-based optimal control functions for mutations rates.

These theoretical results allow us to formulate hypotheses about monotonicity and mutation rate control in biological fitness landscapes.  We test these hypotheses by numerically obtaining optimal mutation rate control functions for 115 published complete landscapes of transcription factor binding \citep{Badis09}.  Our results presented in Section~\ref{sec:TF-landscapes} show that all the optimal mutation rate control functions in these biological landscapes do indeed converge to non-trivial forms consistent with the theory developed here. We also observe differences among optimal mutation rate control functions, variation that relates to variation in the landscapes' monotonic properties.  We conclude in Section~\ref{sec:discussion} by discussing how mutation rate control as considered here may be manifested in living organisms.

\section{Fisher's geometric model of adaptation in Hamming space}
\label{sec:geometry}

In this section, we consider an abstract problem, in which organisms are represented as points in some metric space and adaptation as a motion in this space towards some target point (an optimal organism), and fitness is negative distance to target. Minimisation of distance to the target is therefore equivalent to maximisation of fitness. Geometry of the metric space allows us to solve the optimisation problem precisely.  These abstract results will be used in the following sections to develop the theory further bringing it closer to biology.

\subsection{Representation and assumptions}
\label{sec:representation}

Let $\Omega$ be the set of all possible genotypes representing organisms.  This set is usually equipped with a metric $d:\Omega\times\Omega\to[0,\infty)$ related to the mutation operator, such that large mutations correspond to large distance $d(a,b)$ and vice versa.  For example, the set of all DNA sequences of length $l\in\bN$ can be represented by vectors in the Hamming space $\cH_\alpha^l:=\{1,\ldots,\alpha\}^l$ equipped with the Hamming metric $d_H(a,b)=\sum_{i=1}^l\delta_{a_i}(b_i)$ counting the number of different letters in two strings.  This choice of metric is particularly suitable for a simple point-mutation, which will be the focus of this paper. A sphere $S(a,r)$ and a closed ball $B[a,r]$ of radius $r\in[0,\infty)$ around $a\in\Omega$ are defined as usual:
\[
S(a,r):=\{b\in\Omega:d(a,b)=r\}\,,\quad
B[a,r]:=\{b\in\Omega:d(a,b)\leq r\}
\]
We refer to $r=d(a,b)$ as the \emph{mutation radius}.

Environment defines a preference relation $\lesssim$ (a total pre-order) so that $a\lesssim b$ means genotype $b$ represents an organism that is better adapted to or has a higher replication rate in a given environment than an organism represented by genotype $a$.  We shall consider only countable or even finite $\Omega$, so that there always exists a real function $f:\Omega\rightarrow\bR$ such that
\[
  a\lesssim b\qquad\iff\qquad f(a)\leq f(b)
\]
In game theory, such a function is called utility, but in the biological context it is called \emph{fitness}, and usually it is assumed to have non-negative values representing replication rates of the organisms.  The non-negativity assumption is not essential, however, because the preference relation $\lesssim$ induced by $f$ does not change under a strictly increasing transformation of $f$.  Thus, our interpretation of fitness simply as a numerical representation of a preference relation is distinct from population genetic definitions of fitness \citep[e.g. see][]{Orr09}.  We shall assume also that there exists a top (optimal) genotype $\top\in\Omega$ such that $f(\top)=\sup f(\omega)$, which represents the most adapted or quickly replicating organism.  Note that a finite set $\Omega$ always contains at least one top $\top$ as well as at least one bottom element $\bot$.

Generally, one should consider also the set of all environments (including other organisms), because different environments impose different preference relations on $\Omega$, which have to be represented by different fitness functions.  In this paper, however, we shall assume that fitness in any particular environment has been fixed.

During replication, genotype $a$ can mutate into $b$ with transition probability $P(b\mid a)$.  Mutation can have different effects on fitness: It can be deleterious, if $f(a)>f(b)$; neutral, if $f(a)=f(b)$; or beneficial, if $f(a)<f(b)$.




\begin{figure}[!ht]
\begin{center}
\setlength{\unitlength}{1mm}
\begin{picture}(50,40)
\linethickness{.075mm}
\put(34,17){\circle*{1}}
\put(34,15){$a$}
\put(34,17){\vector(0,1){7}}
\put(34,24){\circle*{1}}
\put(35,24.5){$b$}
\put(34.5,20){$r$}
\put(34,17){\circle{14}}
\put(5,35){\circle*{1}}
\put(2,35){$\top$}
\put(5,35){\vector(3,-2){28}}
\put(18,22){$n$}
\qbezier(5,5)(40,5)(40,40)
\put(5,35){\vector(4,-1){30}}
\put(22,32){$m$}
\qbezier(5,8)(37,8)(37,40)
\end{picture}
\end{center}
\caption{Mutation of point $a$ in a metric space into $b$ with mutation radius $r=d(a,b)$.  The distances $n=d(\top,a)$ and $m=d(\top,b)$ from an optimal element $\top$ define the fitnesses of $a$ and $b$.  }

\label{fig:mutation}
\end{figure}

In this section, we consider a simple picture $f(\omega)=-d(\top,\omega)$, so that maximization of fitness $f(\omega)$ is equivalent to minimization of distance $d(\top,\omega)$, and adaptation (beneficial mutation) corresponds to a transition from a sphere of radius  $n=d(\top,a)$ into a sphere of a smaller radius $m=d(\top,b)$, which is depicted in Figure~\ref{fig:mutation}.  This geometric view of mutation and adaptation is based on Ronald Fisher's idea \citep{Fisher30}, which was, perhaps, the earliest mathematical work on the role of mutation in adaptation.   Fisher represented individual organisms by points of Euclidean space $\bR^l$ of $l\in\bN$ traits, and equipped with the Euclidean metric $d_E(a,b)=(\sum_{i=1}^l |b_i-a_i|^2)^{1/2}$.  The top element $\top$ was identified with the origin in $\bR^l$, and fitness $f(\omega)$ with the negative distance $-d_E(\top,\omega)$.  Then Fisher used the geometry of the Euclidean space to show that the probability of beneficial mutation decreases exponentially as the mutation radius increases, and therefore mutations of small radii are more likely to be beneficial.   Despite subsequent development of the theory \citep{Orr05}, the use of Euclidean space for representation was not revised.

Euclidean space is unbounded (and therefore non-compact) and the interior of any ball has always smaller volume than its exterior.  Therefore, assuming mutation in random directions, a point on the surface of a ball around an optimum is always more likely to mutate into the exterior than the interior of this ball.  This simple property is key for Fisher's conclusion that adaptation is more likely to occur by small mutations.  We showed previously, however, that the geometry of a finite space, such as the Hamming space of strings, implies a different relation between the radius of mutation and adaptation \citep{Belavkin_etal11:_ecal11,Belavkin11:_itw11}.  In particular, the mutation radius maximising the probability of adaptation varies as a function of the distance to the optimum.


\subsection{Probability of adaptation in a Hamming space}




Consider mutation of genotype $a\in S(\top,n)$ in a Hamming space $\cH_\alpha^l$ into $b\in S(\top,m)$ with mutation radius $r=d(a,b)$, as shown on Figure~\ref{fig:mutation}.  Assuming equal probabilities for all points in the sphere $S(a,r)$, the probability that the offspring is in the sphere $S(\top,m)$ is given by the number of points in the intersection of spheres $S(\top,m)$ and $S(a,r)$:
\begin{equation}
P(m\mid n,r)=\frac{|S(\top,m)\cap S(a,r)|_{d(\top,a)=n}}{|S(a,r)|}
\label{eq:p-intersection}
\end{equation}
where $|\cdot|$ denotes cardinality of a set (the number of its elements).  The cardinality of the intersection $S(\top,m)\cap S(a,r)$ with condition $d(\top,a)=n$ is computed as follows:
\begin{eqnarray}
\lefteqn{\left|S(\top,m)\cap S(a,r)\right|_{d(\top,a)=n}}\nonumber\\
&=&\sum_{\substack{r_0+r_-+r_+=\min\{r,m\}\\r_+-r_-=n-\max\{r,m\}}} (\alpha-2)^{r_0}{n-r_+\choose r_0}(\alpha-1)^{r_-}{l-n\choose r_-}{n\choose r_+}\label{eq:h-intersection}
\end{eqnarray}
where summation runs over indexes $r_+\in[0,\lfloor (n-|r-m|)/2\rfloor]$ and $r_-\in[0,\lfloor (r+m-n)/2\rfloor]$ (here $\lfloor\cdot\rfloor$ denotes the floor operation) and satisfying conditions $r_0+r_-+r_+=\min\{r,m\}$ and $r_+-r_-=n-\max\{r,m\}$.  See Appendix~\ref{sec:spheres} for the derivation of this combinatorial result.  We point out that for $r\leq m$, the indexes $r_+$, $r_-$ and $r_0$ count respectively the numbers of beneficial, deleterious and neutral substitutions in $r\in[0,l]$.

The cardinality of sphere $S(a,r)\subset\cH_\alpha^l$ is
\begin{equation}
|S(a,r)|=(\alpha-1)^r{l\choose r}
\label{eq:h-sphere}
\end{equation}
Equations~(\ref{eq:p-intersection})--(\ref{eq:h-sphere}) allow us to compute the probability of adaptation, which is the probability that the offspring is in the interior of ball $B[\top,n]$:
\begin{equation}
P(m<n\mid n,r)=\sum_{m=0}^{n-1}P(m\mid n,r)
\label{eq:p-adaptation}
\end{equation}

\begin{figure}[!t]
\centering
\input{fisher4-k-bw}
\caption{Probability of adaptation $P(m<n\mid n,r)$ in the Hamming
  space ${\cal H}_4^{100}$ as a function of mutation radius $r$.
  Different curves show $P(m<n\mid n,r)$ for different distances
  $n=d_H(\top,a)$ of the parent string from the optimum $\top$.}
\label{fig:fisher-r}
\end{figure}

Figure~\ref{fig:fisher-r} shows the probability of adaptation for Hamming space $\cH_4^{100}$ as a function of mutation radius $r$ for different values of $n=d(\top,a)$.  One can see that when $n<75$ (more generally when $n<l(1-1/\alpha)$), the probabilities of adaptation decrease with increasing radius $r>0$, similar to Fisher's conclusion for the Euclidean space.  However, for $n=75$ there is no such decrease, and when $n>75$ (i.e. for $n>l(1-1/\alpha)$), the probability of adaptation actually increases with $r$.  This is due to the fact that, unlike Euclidean space, Hamming space is finite, and the interior of ball $B[\top,n]$ can be larger than its exterior.  The geometry of a Hamming space has a number of interesting properties \citep{Ahlswede-Katona77}.  For example, every point $\omega$ has $(\alpha-1)^l$ diametric opposite points $\neg\omega$, such that $d_H(\omega,\neg\omega)=l$, and the complement of a ball $B[\omega,r]$ in $\cH_\alpha^l$ is the union of $(\alpha-1)^l$ balls $B[\neg\omega,l-r-1]$.

\subsection{Random mutation}
\label{sec:random-mutation}

By mutation we understand a random process of transforming the parent string $a$ into offspring $b$, so that the mutation radius is a random variable.  The simplest form of mutation, called \emph{point mutation}, is the random process of independently substituting each letter in the parent string $a\in\{1,\ldots,\alpha\}^l$ to any of the other $\alpha-1$ letters with probability $\mu$.  At its simplest, with one parameter, there is an equal probability  $\mu/(\alpha-1)$ of mutating to each of the $\alpha-1$ letters.  The parameter $\mu$ is called the \emph{mutation rate}.  For point mutation, the probability of mutating by radius $r\in[0,l]$ is given by the binomial distribution:
\begin{equation}
P_\mu(r\mid n)={l\choose r}\,\mu^r(n)(1-\mu(n))^{l-r}
\label{eq:p-radius}
\end{equation}
We assume that the mutation rate $\mu$ may depend on the distance $n=d(\top,a)$ from the top string $n$, and therefore the probability is also conditional on $n$.

Optimisation of the mutation rate requires knowledge of the probability $P_\mu(m\mid n)$ that mutation of $a$ into $b$ leads to a transition from sphere $S(\top,n)$ into $S(\top,m)$.  This transition probability can be expressed as follows:
\[
P_\mu(m\mid n)=\sum_{r=0}^l P(m\mid n,r)\,P_\mu(r\mid n)
\]
Substituting~(\ref{eq:p-intersection}) and (\ref{eq:p-radius}) into the above equation, and taking into account~(\ref{eq:h-sphere}), we obtain the following expression:
\begin{equation}
P_\mu(m\mid n)=\sum_{r=0}^l\frac{\left|S(\top,m)\cap S(a,r)\right|_{d(\top,a)=n}}{(\alpha-1)^r}\,\mu^r(n)(1-\mu(n))^{l-r}
\label{eq:p-transition}
\end{equation}
where the number $|S(\top,m)\cap S(a,r)|_{d(\top,a)=n}$ is given by equation~(\ref{eq:h-intersection}).  The case $\alpha=2$ was investigated previously by several authors \citep[e.g.][]{Back93,Braga-Aleksander94}.  The expressions for arbitrary alphabets were first presented in \citep{Belavkin_etal11:_ecal11} \citep[see also][]{Belavkin11:_itw11}.

We note that simple, one parameter point mutation is optimal in a certain sense: it is the solution of a variational problem of minimisation of expected distance between points $a$ and $b$ in a Hamming space subject to a constraint on mutual information between $a$ and $b$ \citep[see][]{Belavkin11:_itw11,Belavkin11:_qbic11}.  The constraint on mutual information between strings $a$ and $b$ represents the fact that perfect copying is not possible.  The optimal solutions to this problem are conditional probabilities having exponential form $P_\beta(b\mid a)\propto\exp[-\beta\,d(a,b)]$, where parameter $\beta>0$, called the \emph{inverse temperature}, is related to the mutation rate, and it is defined from the constraint on mutual information.  The reason why this exponential solution in the Hamming space corresponds to independent substitutions with the same probability $\mu/(\alpha-1)$ is because Hamming metric is computed as the sum $d_H(a,b)=\sum_{i=1}^l\delta_{a_i}(b_i)$ of elementary distances $\delta_{a_i}(b_i)$ between letters $a_i$ and $b_i$ in $i$th position in the string, and the values $\delta_{a_i}(b_i)$ are equal to zero or one independent of the specific letters of the alphabet or their position $i$.  Other, more complex mutation operators, which incorporate multiple parameters or non-independent substitutions (the phenomenon known in biology as \emph{epistasis}) can be considered as optimal solutions of the same variational problem, but applied to a different representation space $\cH$ with a different metric.

\begin{figure}[t]
\centering
\input{mutation-rates-10-4-4-bw}
\caption{Different optimal mutation rate control functions derived mathematically to optimise different criteria in Hamming space $\cH_4^{10}$: Step function minimising expected distance to optimum in one generation, linear function maximising probability of mutating directly into optimum, a function maximising conditional probability $P(m<n\mid n)$ that an offspring is closer to optimum than its parent, cumulative distribution function $P_0(m<n)=\sum_{m=0}^{n-1}{l \choose m}\frac{(\alpha-1)^m}{\alpha^l}$ minimising expected distance to optimum subject to an information constraint \citep{Belavkin11:_dyninf}.}
\label{fig:mutation-rate-functions}
\end{figure}

\subsection{Optimal control of mutation rates}
\label{sec:theory-optimal-control}

The fact that the transition probability $P_\mu(m\mid n)$, defined by equation~(\ref{eq:p-transition}), depends on the mutation rate $\mu$ introduces the possibility of organisms maximising the expected fitness of their offspring by controlling the mutation rate.  We call the collection of pairs $(n,\mu)$ the mutation rate \emph{control function} $\mu(n)$.  Indeed, let $P_t(a)$ be the distribution of parent genotypes in $\cH_\alpha^l$ at time $t$, and let $P_t(n)=\sum_{a:d(\top,a)=n}P_t(a)$ be the distribution of their distances $n=d_H(\top,a)$ from the optimum.  Transition probabilities $P_\mu(m\mid n)$ define a linear transformation  $T_{\mu(n)}(\cdot):=\sum_{n=0}^lP_\mu(m\mid n)(\cdot)$ of distribution $P_t(n)$ into distribution $P_{t+1}(m)$ of distances $m=d_H(\top,b)$ of their offspring at time $t+1$:
\[
P_{t+1}(m)=T_{\mu(n)}P_t(n)=\sum_{n=0}^lP_\mu(m\mid n)P_t(n)
\]
If this transformation does not change with time, then the distribution $P_{t+s}(m)$ after $s$ generations is defined by $T_{\mu(n)}^s$, the $s$th power of $T_{\mu(n)}$.  The optimal mutation rates can be found (at least in principle) by minimising the expected distance subject to additional constraints, such as the time horizon $\lambda$:
\begin{equation}
  \min_{\mu(n)}\ \bE_{P_{t+s}}\{m\}=\sum_{m=0}^lm\,\left(T_{\mu(n)}^sP_t(n)\right)\quad\mbox{subject to $s\leq\lambda$}
  \label{eq:var-problem}
\end{equation}
For example, mutation rates minimising the expected distance at $\lambda=1$ generation should depend on $n$ according to the following step function:
\begin{equation}
\mu(n)=\left\{\begin{array}{ll}
0&\mbox{ if $n<l(1-1/\alpha)$}\\
1-1/\alpha&\mbox{ if $n=l(1-1/\alpha)$}\\
1&\mbox{ if $n>l(1-1/\alpha)$}
\end{array}\right.
\label{eq:step}
\end{equation}
This function is shown on Figure~\ref{fig:mutation-rate-functions} for Hamming space $\cH_4^{10}$.  The sudden change of the optimal mutation rate from $\mu=0$ at $n<l(1-1/\alpha)$ to $\mu=1$ at $n>l(1-1/\alpha)$ corresponds to the sudden change of the effect of the mutation radius on the probability of adaptation shown on Figure~\ref{fig:fisher-r}.  Note that this mutation control function is not optimal for minimisation of the expected distance up to $\lambda>1$ generations, because strings that are closer to the optimum than $l(1-1/\alpha)$ do not mutate, so that there is no chance of improvement.

%
 
The variational problem for the optimal control of the mutation rate, such as problem~(\ref{eq:var-problem}), can be formulated in different ways optimising different criteria (e.g. instantaneous or cumulative expected distance, probability of adaptation, probability of mutating directly into the optimum) or taking into account additional constraints (e.g. the time horizon, information constraints), and generally they lead to different solutions.  Previously, we investigated various types of such problems and obtained their solutions \citep{Belavkin_etal11:_ecal11,Belavkin11:_itw11,Belavkin11:_dyninf}, some of which are shown on Figure~\ref{fig:mutation-rate-functions}.  One can see that there is no single optimal mutation rate control function.  However, it is also evident that all these control functions have a common property of monotonically increasing mutation rate with increasing distance from the optimum.  The main question that we are interested in this paper is whether such monotonic control of mutation rate is beneficial in a broader class of landscapes, when fitness is not equivalent to distance. In Section~\ref{sec:monotonic}, we shall further develop the theory from the simple case considered in this section to more general fitness landscapes and formulate hypotheses which will be tested in biological landscapes in Section~\ref{sec:TF-landscapes}.  To generate data for this testing, we develop an evolutionary technique in Section~\ref{sec:meta-ga} to obtain approximations to the optimal control functions in a broad class of problems, when the derivation of exact solutions is impractical or impossible.


\section{Evolutionary optimisation of mutation rate control functions}
\label{sec:meta-ga}

Analytical approaches cannot always be applied to derive optimal mutation rate control functions due to high problem complexity.  Moreover, when fitness is not equivalent to negative distance, the transition probabilities between fitness levels may be unknown, so that analytical solutions are impossible.  Another approach is to use numerical optimisation to obtain approximately optimal solutions.  In this section, we describe an evolutionary technique that uses two genetic algorithms.  The first, which we refer to as the Inner-GA, evolves individual string with the mutation rate controlled by some function $\mu(y)$ that maps fitness value $y=f(\omega)$ of a string to its mutation rate $\mu\in[0,1]$.  The second, which we refer to as the Meta-GA, evolves a population $\{\mu_1(y),\ldots,\mu_n(y)\}$ of such mutation rate control functions for better performance of the Inner-GAs.  Note that the Inner-GA can use any fitness function.  In this section, we shall apply the technique to the case when fitness is equivalent to negative distance from an optimum (a selected point in a Hamming space).  The purpose of this exercise is to demonstrate that the functions $\mu(y)$ evolved by the Meta-GA have monotonic properties, similar to those possessed by the optimal mutation rate control function obtained analytically.  Later we shall apply the technique to more general fitness landscapes.

\subsection{Inner-GA}
 
The Inner-GA is a simple generational genetic algorithm, where each genotype is a string in Hamming space $\cH_\alpha^l$, and the optimal string is defined by a fitness function $y=f(\omega)$.  The initial population of 100 individuals had equal numbers of individuals at each fitness value, and they were evolved by the Inner-GA for 500 generations using simple point mutation.  The mutation rates were controlled according to function $\mu(y)$, specified by the Meta-GA, with fitness values as the input.  In the experiments described, we used no selection and no recombination in order to isolate the effect on evolution of the mutation rate control from other evolutionary operators.

Note that the parameters of the Inner-GA (e.g. population size, the number of generations) were chosen empirically to satisfy two conflicting objectives:  On one hand, the parameters should be large enough to get any sort of convergence at the Meta-GA level; on the other hand, the parameters should be small enough for the system to obtain satisfactory results in feasible time (in our case several months of run-time using a cluster of 72 GPUs).

\subsection{Meta-GA}

The Meta-GA is a simple generational genetic algorithm that uses tournament selection, which is known to be robust for fitness scores on arbitrary scales and shifts, and because of its suitability for highly parallel implementation.  Each genotype in the Meta-GA is a mutation rate function $\mu(y)$ of fitness values $y$.  The domain of $\mu(y)$ is an ordered partition of the range $\{y:f(\omega)=y,\ \omega\in\cH_\alpha^l\}$ of the Inner-GA fitness function.  Thus, individuals in the Meta-GA are strings of real values $\mu\in[0,1]$ representing probabilities of mutation at different fitnesses, as used in the Inner-GA.

\begin{figure}[t]
\begin{center}
\input{mutation-rates-evolved-10-4-bw}
\caption{Means and standard deviations of mutation rates evolved to minimise expected distance to the optimum in Hamming space $\cH_4^{10}$ after 500 generations.  The results are based on 20 runs of the Meta-GA, each evolving mutation rates for $5\cdot 10^5$ generations.  Each generation of the Meta-GA involved running the Inner-GAs for 500 generations with 100 individuals.  Dashed lines represent theoretical functions optimising short-term (step) and long-term (linear) criteria.}
\label{fig:mutation-10-4-d}
\end{center}
\end{figure}

At each generation of the Meta-GA, multiple copies of the Inner-GA were evolved for 500 generations, with the mutation rate in each copy controlled by a different function $\mu(y)$ taken from the Meta-GA population.  We used populations of 100 individual functions, which were initialised to $\mu(y)=0$.  All runs within the same Meta-GA generation were seeded with the same initial population of the Inner-GA.  The Meta-GA evolved functions $\mu(y)$ for $5\cdot10^5$ generations to maximise the average fitness $\bar y(t)\approx\bE\{y\}(t)$ in the final generation of the Inner-GA.

The Meta-GA used the following selection, recombination and mutation methods:
\begin{itemize}
\item Randomly select (without replacement) three individuals from the population and replace the least fit of these with a mutated crossover of the other two; repeat with the remaining individuals until all individuals from the population have been selected or fewer than three remain.
\item Crossover recombines the start of the numerical string representing one mutation rate function with the end of another using a single cut point chosen randomly, excluding the possibility of being at either end, so that there are no clones.
\item Mutation adds a uniform-random number $\Delta\mu\in[-.1,.1]$ to one randomly selected value $\mu$ (mutation rate) on the individual mutation rate function, but then bounds that value to be within $[0,1]$.
\end{itemize}

The Meta-GA returns the fittest mutation rate function $\mu(y)$.  In this study, the parameters in the Meta-GA were not optimised, as this would probably take more computational time than conducting the study itself.  However, given that Meta-GA converged to the same result, the only difference the parameters could make were how quickly the result was found.

\subsection{Evolved control functions}

The kind of mutation rate control function the Meta-GA evolves depends greatly on properties of the fitness landscape used in the Inner-GA.  In Section~\ref{sec:theory-optimal-control} we showed theoretically that for $f(\omega)$ corresponding to negative distance to optimum $-d_H(\top,\omega)$, the optimal mutation rate increases with $n=d_H(\top,\omega)$.  Therefore, the population of mutation rate functions in the Meta-GA should evolve the same characteristics in such a landscape.  Figures~\ref{fig:mutation-10-4-d} shows the average and standard deviations of the fittest control functions evolved in 20 runs of the Meta-GA using Inner-GAs with strings in $\cH_4^{10}$  (i.e. $\alpha=4$, $l=10$) and fitness defined by $f(\omega)=-d_H(\top,\omega)$.  As predicted, the mutation rate increases with $n=d_H(\top,\omega)$.  We shall now consider more complex landscapes.

\section{Weakly monotonic fitness landscapes}
\label{sec:monotonic}

The derivation of variable (or adaptive) optimal mutation rates described in Section~\ref{sec:theory-optimal-control} was based on the assumption that fitness $f(\omega)$ is equivalent to negative distance $-d(\top,\omega)$ from the top genotype.  Biological fitness landscapes, however, can be rugged \citep{Lobkovsky11}, meaning that fitness may have very little relation to distance in the space of genotypes.  In this section, we consider a more general relation between fitness and distance to study the effect of variable mutation rates on adaptation in more biologically realistic landscapes.  We begin by considering fitness as a noisy or partially observed distance, and then discuss monotonic relation between these ordered random variables.  We introduce several notions of monotonicity and then prove a theorem on weak monotonicity in a general class of fitness landscapes.

\subsection{Fitness-distance communication}

If fitness $y=f(\omega)$ is not isomorphic with distance $n=d(\top,\omega)$, but there is some degree of dependency between the two variables, then one could try to estimate unobserved distance from observed values of fitness and employ the control function $\mu(n)$ of mutation rate based on the estimated distance.  Such a control becomes \emph{$\varepsilon$-optimal}, where $\varepsilon$ represents some deviation from optimality.  The estimation of distance could be done sequentially using, for example, the filtering theory \citep{Stratonovich59:_nonlinear}.  Here, however, we shall limit our discussion to a simple case of using just the current fitness value $y_t$ instead of current distance $n_t$ to control the mutation rate.

Given a distribution $P(\omega)$ of strings in $\Omega$ (e.g. a uniform distribution $P(\omega)=\alpha^{-l}$ on a Hamming space $\cH_\alpha^l$), the fitness $y=f(\omega)$ and distance $n=d(\top,\omega)$ is a pair of random variables with joint distribution $P(y,n)$.  Note that if $\Omega$ has multiple optima $\top$, then $n$ should be understood as the distance from the nearest optimum: $n=\inf\{d(\top,\omega):\top\in\Omega\}$.  Joint distribution $P(y,n)$ defines conditional probabilities $P(n\mid y)$ and $P(y\mid n)$ by the Bayes formula.  Mutation of string $a$ into string $b$ results in the change of distance from $n_t=d(\top,a)$ to $n_{t+1}=d(\top,b)$ and the change of fitness from $y_t=f(a)$ to $y_{t+1}=f(b)$.  If fitness does not communicate more information about the distance than distance itself (i.e. fitness is a `noisy' distance), then one can show that fitness and distance are conditionally independent: $P(y_{t+1},y_t\mid n_{t+1},n_t)=P(y_{t+1}\mid n_{t+1})\,P(y_t\mid n_t)$ (see Remark~\ref{rem:cond-independence} in Appendix~\ref{sec:memoryless}).  In this case, the transition probability $P_\mu(y_{t+1}\mid y_t)$ between fitness values is expressed using the following composition of transition probabilities $P(n_t\mid y_t)$, $P_\mu(n_{t+1}\mid n_t)$ and $P(y_{t+1}\mid n_{t+1})$:
\[
P_\mu(y_{t+1}\mid y_t)=\sum_{n_{t+1}=0}^l\sum_{n_t=0}^l P(y_{t+1}\mid n_{t+1})P_\mu(n_{t+1}\mid n_t)P(n_t\mid y_t)
\]
(see Appendix~\ref{sec:memoryless} for details).  The transition probability $P_\mu(n_{t+1}\mid n_t)$ is defined by the geometry of the mutation operator in the space of genotypes $\Omega$, and for simple point mutation in a Hamming space it is given by equation~(\ref{eq:p-transition}).  Conditional probabilities  $P(n_t\mid y_t)$ and $P(y_{t+1}\mid n_{t+1})$ are defined by the fitness landscape, and they represent dependecy between fitness and distance.

The simplest and, perhaps, the most important such relationship is linear dependency, represented by correlation. The fitness-distance correlation has been used previously to describe problem difficulty for evolutionary algorithms \citep{Jones-Forrest95:_fdc,Jansen2001} and neutral mutations \citep{Poli-Galvan12}.  The fitness-distance correlation reflects global monotonic dependency between the pair of ordered random variables.  In biological context, however, such a global measure of monotonicity may be less important, because biological organisms tend to populate some neighbourhoods of local optima of fitness landscapes due to selection.  Thus, we define the concepts of local and weak monotonicity relative to a chosen metric.  We shall also prove that all landscapes that are continuous and open at local optima are weakly monotonic.  This result will allow us to formulate three hypotheses about control of mutation rates in biological landscapes, which we shall test experimentally in Seciton~\ref{sec:TF-landscapes}.

\subsection{Monotonicity of fitness and distance}

We first consider monotonic relationships between values of fitness function $f:\Omega\to\bR$ at points $a$, $b\in\Omega$ and their distances to an arbitrary point $\omega$ given by a metric $d:\Omega\times\Omega\to\bR_+$ (e.g. a Hamming metric in $\cH_\alpha^l$).  If all $a$ and $b$ inside some ball $B[\omega,n]$, $n>0$, satisfy the properties below, we say that:
\begin{itemize}
\item $f$ is \emph{locally monotonic} relative to metric $d$ at $\omega$ if:
\[
-d(\omega,a)\leq -d(\omega,b)\quad\Longrightarrow\quad f(a)\leq f(b)
\]
\item $d$ is \emph{locally monotonic} relative to $f$ at $\omega$ if:
\[
-d(\omega,a)\leq -d(\omega,b)\quad\Longleftarrow\quad f(a)\leq f(b)
\]
\item $f$ and $d$ are \emph{locally isomorphic} at $\omega$ if both implications hold.
\item We say that $d$ or $f$ are \emph{globally monotonic} (\emph{isomorphic}) at $\top$ relative to each other if the relevant property holds over $B[\omega,l]\equiv\Omega$.
\end{itemize}

\begin{figure}[!ht]
\begin{center}
\setlength{\unitlength}{2mm}
a)
\begin{picture}(12,6)
\linethickness{.2mm}
\put(0,0){\line(2,3){2}}
\put(12,0){\line(-2,3){2}}
\put(2,3){\line(1,0){2}}
\put(10,3){\line(-1,0){2}}
\put(4,3){\line(2,3){2}}
\put(8,3){\line(-2,3){2}}
\end{picture}
b)
\begin{picture}(12,6)
\linethickness{.2mm}
\put(0,0){\line(3,2){3}}
\put(12,0){\line(-3,2){3}}
\put(3,2){\line(0,1){2}}
\put(9,2){\line(0,1){2}}
\put(3,4){\line(3,2){3}}
\put(9,4){\line(-3,2){3}}
\end{picture}
c)
\begin{picture}(12,6)
\linethickness{.2mm}
\put(0,0){\line(1,1){6}}
\put(12,0){\line(-1,1){6}}
\end{picture}
\end{center}
\caption{Schematic representation of monotonic properties described.  Abscissae represent string space, ordinates represent fitness. a) Fitness is monotonic relative to distance to optimum (fitness landscape can have `plateaus'); b) distance to optimum is monotonic relative to fitness (landscape can have `cliffs'); c) fitness and distance to optimum are isomorphic (neither cliffs nor plateaus are allowed).}
\label{fig:monotonicity}
\end{figure}

The three monotonic relations between fitness and distance defined above are illustrated on Figure~\ref{fig:monotonicity}.  These cases represent idealised situations, because usually the value of distance does not define the value of fitness uniquely and vice versa.  Indeed, the pre-image of distance $d(\top,\omega)=n$ is a sphere $S(\top,n):=\{\omega:d(\top,\omega)=n\}$, and there can be strings with different fitness values within the sphere.  Similarly, the pre-image of fitness $f(\omega)=y$ is the set $f^{-1}(y)=\{\omega:f(\omega)=y\}$, and strings within this set may have different distances from $\top$.  Thus, to describe monotonicity in realistic landscapes, one can modify the definitions by considering the `average' (i.e. expected) fitness or distance within the sets.  In particular, we shall denote the average fitness at distance $d(\top,\omega)=n$ and the average distance at fitness $f(\omega)=y$ respectively as follows:
\[
\bE[f(\omega_n)]:=\frac{1}{|S(\top,n)|}\sum_{\omega\in S(\top,n)}f(\omega)\,,\ 
\bE[d(\top,\omega_y)]:=\frac{1}{|f^{-1}(y)|}\sum_{\omega\in f^{-1}(y)}d(\top,\omega)
\]
The above averages are particular cases of conditional expectations under the assumption of a uniform distribution $P(\omega)=\alpha^{-l}$ of strings in a Hamming space $\cH_\alpha^l$.  Appendix~\ref{sec:weak-mon} gives definitions for arbitrary Borel probability measure on a metric space.

In what follows, we shall use specific notation $\bE[f(a)]$, with letter `$a$' instead of $\omega_n$, to denote average fitness at distance $d(\top,\omega)=d(\top,a)$ (instead of $d(\top,\omega)=n$).  Similarly, we shall use notation $\bE[d(\top,a)]$, with a specific letter `$a$' instead of $\omega_y$, to denote average distance at fitness $f(\omega)=f(a)$ (instead of $f(\omega)=y$).  Such notation is convenient to define average (mean) monotonicity.  If all $a$ and $b$ inside some ball $B[\omega,n]$, $n>0$, satisfy the properties below, we say that:
\begin{itemize}
\item $f$ is on average \emph{locally monotonic} relative to metric $d$ at $\omega$ if:
\[
-d(\omega,a)\leq -d(\omega,b)\quad\Longrightarrow\quad \bE[f(a)]\leq \bE[f(b)]
\]
\item $d$ is on average \emph{locally monotonic} relative to $f$ at $\omega$ if:
\[
-\bE[d(\omega,a)]\leq -\bE[d(\omega,b)]\quad\Longleftarrow\quad f(a)\leq f(b)
\]
\item $f$ and $d$ are on average \emph{locally isomorphic} at $\omega$ if both implications hold.
\end{itemize}

If $\omega$ is a local optimum in a finite space (e.g. a Hamming space), then fitness and distance are always locally monotonic relative to each other (and hence isomorphic) at least inside the ball $B[\omega,1]$ of radius one (otherwise, $\omega$ cannot be a local optimum).  However, if the size $|\Omega|$ of the space is large, then the neighbourhood becomes negligible, and therefore the notions of local monotonicity become less important in larger landscapes.  For larger neighbourhoods one can speak only about the probability that the implications above are true for some pair of points $a$ and $b$.  Thus, for larger neighbourhoods we can define monotonicity \emph{in probability} (or in measure).  However, because monotonicity always holds with some (possibly zero) probability, and it holds trivially with probability one at each point (i.e. in a zero-radius ball $B[\omega,0]$), we should make such a definition more useful by distinguishing, for example, landscapes, in which the probability of monotonicity gradually increases as points get closer to a local optimum.  We refer to this notion as \emph{weak monotonicity}.

Let $\{\omega_n\}_{n\in\bN}$ be a sequence of points such that the distances $d(\omega,\omega_n)$ converge to zero or fitness values $f(\omega_n)$ converge to $f(\omega)$.  If $a$, $b\in\{\omega_n\}_{n\in\bN}$ of any such sequence satisfy the properties below, we say that:
\begin{itemize}
\item $f$ is \emph{weakly monotonic} relative to metric $d$ at $\omega$ if:
\[
\lim_{d(\omega,b)\to0} P\Bigl\{-d(\omega,a)\leq -d(\omega,b)\quad\Longrightarrow\quad \bE[f(a)]\leq \bE[f(b)]\Bigr\}=1
\]
\item $d$ is \emph{weakly monotonic} relative to $f$ at $\omega$ if:
\[
\lim_{f(b)\to f(\omega)} P\Bigl\{-\bE[d(\omega,a)]\leq -\bE[d(\omega,b)]\quad\Longleftarrow\quad f(a)\leq f(b)\Bigr\}=1
\]
\item $f$ and $d$ are \emph{weakly isomorphic} at $\omega$ if both conditions hold.
\end{itemize}

Weak monotonicity is implied by the average local monotonicity in some ball $B[\omega,n]$ with $n>0$, because the latter means that the implications above hold with probability one in $B[\omega,n]$.  The average local monotonicity is in turn implied by the (strong) local monotonicity.  The relation between the three notions is shown by the implications below:
\[
\mbox{Local monotonicity}\ \Rightarrow\ 
\mbox{Average local monotonicity}\ \Rightarrow\ 
\mbox{Weak monotonicity}
\]
Moreover, one may consider an increasing sequence of finite landscapes such that in the limit $|\Omega|\to\infty$ the landscape is modelled by a continuum metric space.  In this case, fitness may fail to be monotonic at any point, including the global optimum, even if fitness is a continuous function.  Indeed, it is well known that almost all continuous functions are nowhere differentiable, and therefore they are also nowhere monotonic \citep{Banach1931,Mazurkiewicz1931}.  However, as will be shown by the theorem below, weaker monotonicity may still hold in such landscapes.

Like weak monotonicity, fitness-distance correlation can also be applied to infinite landscapes, inlcuding nowhere monotonic landscapes.  However, while fitness-distance correlation describes global property of a landscape, weak monotonicity effectively describes a gradual increase of fitness-distance correlation in decreasing neighbourhoods of a point.  Thus, although weak monotonicity is related to fitness-distance correlation, these notions are not equivalent.  In fact, unlike fitness-distance correlation, weak monotonicity holds in a very broad class of landscapes, including infinite landscapes.

\begin{theorem}
Let $(\Omega,d)$ be a metric space equipped with a Borel probability measure $P$, and let $f:\Omega\to\bR$ be $P$-measurable.  Let $\top$ be a local optimum: $f(\top)=\sup\{f(\omega):\omega\in E\subseteq \Omega\}$.  Then
\begin{description}
\item [($\Rightarrow$)] If $f$ is continuous at $\top$, then $f$ is weakly monotonic relative to $d$ at $\top$.
\item [($\Leftarrow$)] If $f$ maps open balls $B[\top,\delta)\subseteq E$ to open intervals $(f(\top)-\varepsilon,f(\top)]$, then $d$ is weakly monotonic relative to $f$ at $\top$.
\item [($\iff$)] If $f$ satisfies both conditions above, then $f$ and $d$ are weakly isomorphic at $\top$.
\end{description}
\label{th:weak-mon}
\end{theorem}

The proof of this theorem is given in Appendix~\ref{sec:weak-mon}, and it is based on the construction of a decreasing sequence $\{\delta_n\}_{n\in\bN}$ of radii $\delta_n>0$ around $\top$ for any increasing sequence $\{f(\top)-\varepsilon_n\}_{n\in\bN}$, which is guaranteed by continuity of $f$ at $\top$.  Note that we used metric in the theorem, because metric spaces are well-understood, but the theorem and its proof can be reformulated in terms of a quasi-pseudometric.  Every quasi-uniform space with countable base (and hence every corresponding topological space) is quasi-pseudometrisable \citep[e.g. see][Theorem~1.5]{Fletcher-Lindgren82}, which probably subsumes any topology on DNA or RNA structures \citep{Stadler_etal2001}.

Weak monotonicity implies increasing probability of positive correlation between fitness and negative distance to a local or global optimum in decreasing neighbourhoods.  This suggests that the fitness-based control $\mu(y_t)$ of mutation rate in any continuous and open landscape should resemble the distance-based control $\mu(n_t)$ in some neighbourhood of an optimum. This forms our first hypothesis:
\begin{hypothesis}
Optimal mutation rate increases with a decrease in fitness in some neighbourhood of an optimum for realistic fitness landscapes (e.g. biological landscapes), where fitness is not globally isomorphic to distance.
\label{h1}
\end{hypothesis}
Further, the more monotonic the landscape, the more the optimal mutation rate control function will resemble theoretical functions derived and discussed in Section~\ref{sec:geometry}; this forms our second hypothesis:
\begin{hypothesis}
The larger the neighbourhood of weak monotonicity, the more mutation rate control may contribute to evolution towards high fitness.
\label{h2}
\end{hypothesis}
We test these hypotheses in Section~\ref{sec:TF-landscapes}.

\subsection{On the role of genotype-phenotype mapping}

Mutation occurs at the microscopic level as a random change of a genotype, whereas fitness is defined by the interaction of an organism with its environment, and therefore is a property of the phenotype rather than genotype.  If we denote by $X$ the set of all phenotypes, then fitness of genotypes $f:\Omega\to\bR$ can be factorised into a composition $f=\varphi\circ\kappa$ of a genotype-phenotype mapping $\kappa:\Omega\to X$ and phenotypic fitness $\varphi:X\to\bR$.  We use a function $\kappa$ for genotype-phenotype mapping, because we assume for simplicity that one genotype cannot be decoded into two or more phenotypes.  On the other hand, there are usually many genotypes corresponding to the same phenotype \citep{Schuster_etal94}.  The genotype-phenotype mapping $\kappa$ can be seen as a black-box model of DNA decoding via translation and transcription.

The set $X$ of phenotypes is pre-ordered by the values of phenotypic fitness ($x\lesssim_X z$ iff $\varphi(x)\leq \varphi(z)$), while the set $\Omega$ of genotypes is pre-ordered by the values of distance from the nearest top genotype ($a\lesssim_\Omega b$ iff $-d(\top,a)\leq-d(\top,b)$).  It is clear from factorisation $f=\varphi\circ\kappa$ that the relation between fitness $f$ of genotypes and their distance from an optimum depends on monotonic properties of the genotype-phenotype mapping.  For example, genotypic fitness is order-isomorphic with distance when the genotype-phenotype mapping satisfies the condition: $a\lesssim_\Omega b$ if and only if $\kappa(a)\lesssim_X\kappa(b)$.

The factorisation $f=\varphi\circ\kappa$ shows that part of the fitness function, specifically $\kappa$, is property of an organism, and therefore a monotonic relation between fitness and distance can be an adaptive and evolving property.  This forms our third hypothesis:
\begin{hypothesis}
The extent to which mutation rate control may contribute to the evolution of high fitness is itself a trait, which will evolve across biological organisms.
\label{h3}
\end{hypothesis}
We analyse data that may support this hypothesis in Section~\ref{sec:TF-landscapes}.

\section{Evolving fitness-based mutation rate control functions}
\label{sec:TF-landscapes}



In this section, we conduct a computational experiment using landscapes with biological origins to test the hypotheses arising from our theory in Section~\ref{sec:monotonic}.  We used the earlier described Meta-GA technique (see Section~\ref{sec:meta-ga}) to evolve approximately optimal functions for 115 published complete landscapes of transcription factor binding \citep{Badis09}.  This also allows us to establish the range of fitness values over which monotonicity of optimal mutation rate holds, quantifying the extent to which Hypothesis~\ref{h1} holds for these biological landscapes.  TFs have evolved over very long periods to bind to specific DNA sequences.  The landscapes show experimentally measured strengths of interaction (DNA-TF binding score) between the double-stranded DNA sequences of length $l=8$ of base pairs each and a particular transcription factor.  Thus, we represent the set of all DNA sequences by Hamming space $\cH_4^8$ (i.e. $\alpha=4$, $l=8$), and consider the DNA-TF binding score as their fitness, which is clearly different from the negative Hamming distance from the top string (a sequence with the maximum DNA-TF binding score).


\subsection{Evolved control functions}

We used the Meta-GA evolutionary optimisation technique, described in Section~\ref{sec:meta-ga}, to obtain for each landscape an approximately optimal mutation rate control function maximising the average DNA-TF binding score in the population (expected fitness) after 500 replications.  Our experiments showed that 16 replicate runs\footnote{We used a multiple of 4 due to 4 GPUs used in one node.} were sufficient to achieve satisfactory convergence in feasible time for each of the 115 transcription factor landscapes.

\begin{figure}[t]
\input{landscape-types-mutation-bw}
\caption{Examples of GA-evolved optimal mutation rate control functions.  Data are shown for the transcription factors Srf, Glis2 and Zfp740.  Each curve represents the average of 16 independently evolved optimal mutation rate functions on a particular transcription factor DNA-binding landscape \citep{Badis09}.  Errorbars represent standard deviations from the mean.  Similar curves for all 115 landscapes are shown in supplementary Fig.~\ref{fig:all-curves}.  The arrows indicate the \emph{monotonicity radius} $\varepsilon$, that defines an interval of fitness values below the maximum, where mutation rate monotonically increases.}
\label{fig:three-curves}
\end{figure}

Figure~\ref{fig:three-curves} shows the average values and standard deviations of the evolved mutation rates for three transcription factors: Srf, Glis2 and Zfp740.  Evolved functions for all landscapes are shown on Figure~\ref{fig:all-curves} in supplementary material.  One can see that the evolved functions for each transcription factor landscape is approximately monotonic in the direction predicted: close to zero mutation at the maximum fitness, rising to high levels further from the maximum fitness value.  This supports Hypothesis~\ref{h1} as developed from the theory in Section~\ref{sec:monotonic}.

Small standard deviations indicate good convergence to a particular control function.  Observe that there is poor convergence at low fitness areas of the landscape that are poorly explored by the genetic algorithm.  Once the mutation rate has peaked near the maximum value $\mu=1$, the mutation rates tend to decrease and become chaotic.  As will be shown in the next section, this occurs at lower fitness values at which the landscape is no longer monotonic (i.e. further from the peak of fitness).

\subsection{Landscapes for transcription factors}

The variation in the evolved mutation rate control function is clearly related to a variation in the properties of the landscapes.  Our theoretical analysis suggests that the main property affecting mutation rate control is monotonicity of the landscape relative to a metric measuring the mutation radius.  In particular, the radius of point-mutation is measured by the Hamming metric, and we shall look into the local and weak monotonic properties of the transcription factors landscapes relative to the Hamming metric.

\begin{figure}[t]
\input{landscape-types-fitness-bw}
\caption{Examples of fitness landscapes based on the binding score between DNA sequences and transcription factors (TF) from \citep{Badis09}.  Data are shown for the transcription factors: Srf, Glis2 and Zfp740.  Lines connect mean values of the binding score shown as functions of the Hamming distance from the top string (a sequence with the highest DNA-TF binding score).  Errorbars represent standard deviations.  Similar curves for all 115 landscapes are shown in supplementary Fig.~\ref{fig:all-fitness-local}.
}
\label{fig:landscape-types-fitness}
\end{figure}

Figure~\ref{fig:landscape-types-fitness} shows average DNA-TF binding scores within spheres $S(\top,n)$ around the optimal string as a function of Hamming distance $n=d_H(\top,\omega)$ from the optimum.  Data is shown for three transcription factors: Srf, Glis2 and Zfp740.  Lines connect average values at discrete distances for visualisation purposes.  Errorbars show standard deviations of the DNA-TF binding scores within the spheres.  Distributions of fitness with respect to Hamming distance $d_H(\top,\omega)$ for all 115 transcription factors are shown on Figure~\ref{fig:all-fitness-local} (supplementary material).

One can see from Figure~\ref{fig:landscape-types-fitness} that the landscape for the Srf factor has monotonic properties: The average values increase steadily for strings that are closer to the optimum, and the deviations from the mean within the spheres are relatively small.  This is in contrast to the other two landscapes.  We note also that the average values for Glis2 decrease quite sharply around the optimum, while the landscape for Zfp740 has a relatively flat plateau area around the optimum, which means that there are many sequences with high DNA-TF binding score.  This difference may explain different gradients of optimal mutation rates near the maximum fitness shown on Figure~\ref{fig:three-curves}.

\begin{figure}[t]
\input{edge-tau-rnd-bw}
\caption{Linear relation between monotonicity of the landscapes measured by the Kendall's $\tau$ correlation (ordinates) and the monotonicity radius $\varepsilon$ (abscissae) of the corresponding evolved mutation rate control functions.  Three labels show data for three transcription factors shown in Figs.~\ref{fig:three-curves} and \ref{fig:landscape-types-fitness}.}
\label{fig:edge-tau}
\end{figure}

\subsection{Monotonicity and controllability}

Our results have confirmed that the evolved optimal mutation rates rise from zero to very high levels as fitness decreases from the maximum value $f(\top)$ to some value $f(\top)-\varepsilon$ (see Fig.~\ref{fig:three-curves} and supplementary Fig.~\ref{fig:all-curves}).  We refer to the corresponding value $\varepsilon>0$ as the \emph{monotonicity radius}, as it defines the neighbourhood of $\top$ in terms of fitness values in which the evolved mutation rate control function has monotonic properties.  We find substantial variation in monotonicity radius among transcription factors.

We hypothesised that the variation in the optimal mutation rate control functions relates to variation in the monotonicity of the transcription factor landscapes (Hypothesis~\ref{h2}).  Various measures have been proposed for the roughness of biological landscapes \citep{Lobkovsky11}.  Here we focus on Kendall's $\tau$ correlation, which is directly concerned with monotonicity; specifically, $\tau$ measures the proportion of mutations that, in moving closer to the optimum in string space, also increase in fitness.  As shown in Figure~\ref{fig:edge-tau}, we find that the value of $\tau$ of the landscape does indeed have a relationship with the monotonicity radius $\varepsilon$ of the evolved mutation rate control functions (Spearman's $\rho= 0.77$, $P \approx 10^{-16}$, $N=115$), supporting Hypothesis~\ref{h2}.

Finally, we investigated whether these related features of the TF landscapes and mutation rate functions themselves relate to the biological evolution of these TF systems.  To test this we looked at the evolutionary origins of the TF families, to which the 115 TFs tested above belonged, using an integer scale indicating key splits in the tree of eukaryotic life \citep{Weirauch11}. We find a significant relationship between this scale of biological evolution and the monotonicity radius $\varepsilon$ (Spearman's $\rho = 0.21$, $P = 0.021$, $N = 115$).  This indicates that TFs in families that originated more recently (e.g. in families restricted to Deuterostomes, rather than being present across all eukaryotic life) tend to have broader regions over which the optimal mutation rate monotonically increases with distance from the binding optimum.  This is consistent with Hypothesis~\ref{h3}, indicating that the extent to which mutation rate control may contribute to the evolution of high fitness itself evolves through the tree of life.


\section{Discussion}
\label{sec:discussion}

In this paper we have developed and tested theory relating to the control of the mutation rate in biological sequence landscapes. To do so, we had to move the theory closer to the biology in three ways. Firstly (in Section~\ref{sec:geometry}), we generalised Fisher's geometric model of adaptation, from its Euclidean space (continuous and infinite) to a discrete, finite Hamming space of strings.  Doing so demonstrated that, in contrast to the behaviour in Euclidean space, where the probability of beneficial mutation behaves similarly at different distances from the optimum  \citep{Orr03}, the probability of beneficial mutation, for a given mutation size, varies markedly depending on the distance from the optimum (Figure~\ref{fig:fisher-r}). Secondly, we analytically derived functions for optimal control of the mutation rate minimising the expected Hamming distance to a particular point (optimal string).  We also demonstrated a variation of these control functions dependent on specific formulations of the optimisation problem.  Nonetheless we observed consistency: all optimal functions increase monotonically (Figure~\ref{fig:mutation-rate-functions}). Thirdly, we developed theory concerning monotonic properties of fitness landscapes and establishing sufficient conditions of weak monotonicity.  The theory demonstrated that all biological landscapes over discrete spaces, however rugged, are characterised by monotonic properties in some neighbourhood of the optimum.  Therefore, optimal solutions to the geometric problem of optimal mutation rate control based on distance can be applied more broadly to problems of $\varepsilon$-optimal control of mutation rate based on fitness in biological systems.

Empirical biological fitness landscapes mapping genotypes to fitness values within a small, defined, subset of genotypic space are becoming increasingly available \citep{deVisser14}. Here we use the test case of the affinities of 115 different transcription factors for all possible eight base-pair DNA sequences \citep{Badis09}. We used these landscapes to test hypotheses arising from the theory, relating to the nature of optimal mutation rate functions (Hypothesis~\ref{h1} and Figs.~\ref{fig:three-curves}, \ref{fig:landscape-types-fitness} and \ref{fig:edge-tau}).  In each case we find evidence to support the hypothesis, consistent with the idea that our theory is not only correct, but, as expected, substantively relevant to such biological fitness landscapes.  


Given that we find this theory to be relevant to biological fitness landscapes, we need to ask how it might manifest itself within biology. There are several requirements if biological organisms are to exert any approximation to optimal mutation rate control. The first requirement is variation in mutation rate. There is evidence for abundant variation in biological mutation rates, both across species \citep{Sung12} and among populations of a species \citep{Bjedov03}. Variation is therefore possible.  However, for this theory to be relevant, that variation needs to be controllable by the organism.  This in turn requires that mutation rate varies right down to the level of an individual genotype, i.e. \emph{mutation rate plasticity} (MRP). There is evidence for MRP in `stress-induced mutagenesis' \citep{Galhardo07} and related phenomena, such as the increased number of mutations in sperm from older males \citep{Kong12}. However, while this constitutes MRP, the possibility of control requires that this plasticity is not merely the inevitable result of an organism's environment (e.g. the accumulation of damage with time or due to stress factors), but controllable by the organism in response to that environment. The proximate and ultimate causes of stress-induced mutagenesis are much debated, but that they include any form of `control' is far from clear \citep{Maclean13}. Clearer evidence of control is, however, present in a novel example of MRP we described recently \citep{eids_nature14}. In this case, there is environmentally dependent MRP that can be switched on or off by the presence or absence of a particular gene (\emph{luxS}). 

The next requirement for a biological analogue of the theory described here is that control of the mutation rate may be exercised as a decreasing function of fitness. This requires that an organism can somehow assay its own fitness. This is a non-trivial requirement in that fitness is a function of one or more generations of an organism's offspring, not of an organism itself. Various proxies are conceivable that might give an organism an indication of its fitness. These include counting its offspring relative to some internal or external clock, counting the population as a whole, or testing aspects of the environment that may correlate with the future likelihood of offspring. The last of these could include stressors, meaning that stress-induced mutagenesis might meet this requirement. In our recently identified example, the aspect of the environment with which mutation rate varies is the density of a bacterial culture. Population density can act as a good proxy for fitness in some circumstances (e.g. in a fixed volume bacterial culture), and the mutation rate does indeed decrease with increasing density \citep{eids_nature14}, consistent with the fitness-associated control of mutation rate we here determine to be optimal.

The final requirement for the existence of biological mutation-rate control of the sort addressed here is that it is possible for it to evolve and be maintained by the processes of biological evolution. This is not trivial in that it involves the evolution of plasticity, which is not as straight-forward or common in biology as might be expected \citep{Scheiner12}. It also involves so-called `second-order selection' \citep{Tenaillon01}.  This is because any particular mutation rate or MRP is unlikely to affect an individual's fitness (and therefore selection) directly; rather, MRP must be selected for indirectly via the genetic effects it produces. Nonetheless, phenotypic plasticity occurs widely and, while rare, there are clear examples of second-order selection occurring in biology \citep{Woods11}. Furthermore, here we demonstrate MRP rapidly evolving \emph{de novo} to particular forms (Figure~\ref{fig:three-curves}). The genetic algorithm (GA) in this case was not created to mimic biology, and the group-selection used by the outer GA in particular is rather un-biological. However, others, working with explicitly biological population genetic models, also find the evolution of MRP \citep{Ram12}. This implies that not only is the MRP predicted here possible for biological organisms, but it may reasonably be expected to evolve and be maintained. It remains to be tested whether the precise range and nature of the MRP identified by \citet{eids_nature14} does indeed fulfil this role i.e. to enable populations to evolve faster and/or further in realised, whole organism biological fitness landscapes in a similar fashion to the evolutionary advantage seen for \emph{in silico}, molecular interaction landscapes tested here (Figure~\ref{fig:three-curves}). Nonetheless, such density-dependent MRP \citep{Krasovec14b} is a prime candidate for a biological manifestation of the mutation rate control which we have addressed here.

We have focused on fitness-associated control of mutation rate. However, mutation is only one evolutionary process where fitness-associated control may be beneficial.  Recombination and dispersal are also evolutionary processes that may be under the control of the individual and therefore open to similar effects.  Fitness-associated recombination has been demonstrated to be advantageous theoretically \citep{Hadany03,Agrawal05} and identified in biology \citep{Agrawal08,Zhong11}.  Similarly, the idea that dispersal associated with low fitness might be advantageous has a basis in simulation of spatially differentiated populations  \citep{Aktipis2004,Aktipis2011}. This association might perhaps be framed more generally in terms of `fitness-associated dispersal'.  Thus, the framework for control of mutation rate in response to fitness that we have developed here may in future be applicable to both recombination and dispersal.

Overall, our development of theory and testing its predictions \emph{in silico} not only clarifies ideas around the monotonicity of fitness landscapes and mutation rate control, it leads directly to hypotheses about specific systems in living organisms. At the same time there is the potential for greater insight through further development of the theory.  Three directions seem particularly likely to be fruitful.

First, while it is striking how effective mutation rate control is at enabling adaptive evolution, without invoking selection in our  \emph{in silico} experiments, it will be important to consider the role of selection strategies.  Such strategies may implicitly modify fitness functions.  For instance, one of the analytically derived functions shown in Figure~\ref{fig:mutation-rate-functions} is the mutation rate function for a DNA space ($\cH_4^{10}$) which maximises the probability of adaptation (as derived by \citet{Back93} for binary strings).  As outlined in Section \ref{sec:theory-optimal-control}, maximising the probability of adaptation is equivalent to maximising expected fitness of the offspring relative to its parent. This effect may be implicit in a selection strategy that removes the offspring of reduced fitness that will inevitably be produced by maximising offspring expected fitness. Given the importance of selection in biology, we therefore anticipate that such functions may be closer to mutation rate control functions in living organisms.  This requires further work.

A second area for development is in variable adaptive landscapes.  The importance of time-varying adaptive landscapes in biological evolution is becoming increasingly appreciated \citep{Mustonon09,Collins11} and variable mutation rates have a particular role here \citep{Stich10}.  It is worth noticing, however, that our derivation of optimal mutation rate functions is \emph{not} dependent on a fixed landscape, as it depends only on the fitness values.  Nonetheless, as we demonstrate for the transcription factor landscapes, variation in landscapes' monotonic properties relates to the shape of mutation rate functions in predictable ways (Figure~\ref{fig:edge-tau}).  This deserves further exploration both theoretically and empirically: measuring variation in the monotonic properties of real biological landscapes will be informative about optimal mutation rate functions and \emph{vice versa}.

Finally, there is potential to develop theory around the role of the genotype-phenotype mapping.  Landscape monotonicity, as explored here, is not absolute; it may depend on this mapping.  That is, if the decoding of DNA changes, it may be possible to convert a non-monotonic landscape into a monotonic one.  Biology uses a variety of such decoding schemes which may themselves evolve.  For the transcription factor landscapes used here, the decoding scheme is defined by the biochemical interactions between the transcription factor (a protein molecule) and DNA.  Thus, evolution of transcription factors constitutes evolution of DNA-decoding, and indeed we do find a relationship between the evolutionary age of gene families and the monotonic properties of the associated landscapes.  A more familiar example is the genetic code, where there is much existing work on its evolution \citep[e.g.][]{Freeland00}.  Determining how evolution of such codes affects the monotonic properties of biological landscapes as explored here may, therefore, provide novel insights into large-scale evolutionary patterns.  Ultimately, theory such as this that identifies analytically or empirically optimal mutation rate control functions may help make predictions about evolutionary responses to future environmental change \citep{Chevin10} or inferences about the environment(s) within which particular organisms evolved.  In the meantime, mutation rate control as developed here may assist directed evolution within biological and other complex landscapes, for instance in the evolution of DNA-protein binding \citep{Knight09}.

\bibliographystyle{spbasic}      
\bibliography{rvb,nn,other,newbib,ica,evolution,eids-pointmutation}


\appendix

\section{Intersection of spheres in Hamming space $\cH_\alpha^l$}
\label{sec:spheres}

\begin{proposition}
Let $S(a,r)$ and $S(c,m)$ be two spheres in $\cH_\alpha^l:=\{1,\ldots,\alpha\}^l$ with Hamming distance between the centres $d_H(a,c)=n$.  Then the cardinality of their intersection is
\begin{eqnarray}
\lefteqn{\left|S(a,r)\cap S(c,m)\right|_{d(a,c)=n}}\label{eq:h-intersection-a}\\
&=&\sum_{\substack{%
r_+\in[0,\lfloor (n-|r-m|)/2\rfloor]\\
r_-\in[0,\lfloor (r+m-n)/2\rfloor]\\
r_0+r_-+r_+=\min\{r,m\}\\
r_+-r_-=n-\max\{r,m\}}} (\alpha-2)^{r_0}{n-r_+\choose r_0}(\alpha-1)^{r_-}{l-n\choose r_-}{n\choose r_+}\nonumber
\end{eqnarray}
\label{pr:spheres}
\end{proposition}

\begin{proof}
The intersection is formed by points $b\in S(a,r)\cap S(c,m)\subset\cH_\alpha^l$, which form triangles together with the centres of the spheres $a$ and $c$:
\[
\xymatrix{&b&&\\a\ar[ru]^r\ar@{-}[rrr]_n&&&c\ar[llu]_m}
\]
Point $b$ can be obtained equivalently by substituting $r$ letters in $a$ or by substituting $m$ letters in $c$, and we shall count using the smallest number of substitutions $\min\{r,m\}$.

Consider a substitution of a letter $a_i$ of string $a$ into letter $b_i$.  There are three possible cases:
\begin{description}
\item [$a_i\neq c_i$, $b_i=c_i$]  Such substitutions contribute to a decrease of distance between the strings, and we refer to them as \emph{beneficial} substitutions.  There are $d_H(a,c)=n$ of letters $a_i$ in $a$ such that $a_i\neq c_i$.  The number of $r_+\in[0,n]$ beneficial substitutions out of $n$ letters in $a$ is ${n\choose r_+}$.
\item [$a_i=c_i$, $b_i\neq c_i$]  Such substitutions contribute to an increase of distance between the strings, and we refer to them as \emph{deleterious} substitutions.  There are $l-d_H(a,c)=l-n$ of letters $a_i$ in $a$ such that $a_i=c_i$, and each of them can be substituted into $\alpha-1$ letters $b_i\neq c_i$.  The number of $r_-\in[0,l-n]$ deleterious substitutions out of $l-n$ letters in $a$ is $(\alpha-1)^{r_-}{l-n\choose r_-}$.
\item [$a_i\neq c_i$, $b_i\neq c_i$]  Such substitutions do not change the distance between the strings, and we refer to them as \emph{neutral} substitutions.  After $r_+$ beneficial substitutions, there are $d_H(a,c)-r_+=n-r_+$ of letters $a_i$ in $a$ such that $a_i\neq c_i$, and each of them can be substituted into $\alpha-2$ letters $b_i\neq c_i$.  The number of $r_0\in[0,n-r_+]$ neutral substitutions out of $n-r_+$ letters in $a$ is $(\alpha-2)^{r_0}{n-r_+\choose r_0}$.
\end{description}
The product of these three numbers gives the total number of $r_+$ beneficial, $r_-$ deleterious and $r_0$ neutral substitutions.  For fixed points $a$ and $c$ with $d_H(a,c)=n$, the third point $b\in S(a,r)\cap S(c,m)$ can be obtained using different values of $r_+$, $r_-$ and $r_0$.  However, not all values of $r_+$, $r_-$ and $r_0$ are admissible.

First, let us establish the range for $r_+$ and $r_-$.  Using the triangle inequalities for Hamming metric we have $d_H(a,c)\leq d_H(a,b)+d_H(b,c)$ and $|d_H(a,b)-d_H(b,c)|\leq d_H(a,c)$, or (substituting the values $n$, $m$ and $r$):
\[
|r-m|\leq n\leq r+m
\]
This gives inequalities $r+m-n\geq0$ and $n-|r-m|\geq0$.  Observe that
\[
(r+m-n)+(n-|r-m|)=2\min\{r,m\}
\]

If $r+m-n=0$, then $n=m+r$, which means that all $\min\{r,m\}$ substitutions are beneficial.  Thus, $r_+$ is bounded above by $(n-|r-m|)/2$.

If $n-|r-m|=0$, then $\max\{r,m\}=n+\min\{r,m\}$, which means that all $\min\{r,m\}$ substitutions are deleterious.  Thus, $r_-$ is bounded above by $(r+m-n)/2$.

Given $r_+\in[0,\lfloor (n-|r-m|)/2\rfloor]$ and $r_-\in[0,\lfloor (r+m-n)/2\rfloor]$, the number $r_0$ of neural substitutions is computed from the condition:
\[
r_++r_-+r_0=\min\{r,m\}
\]
Finally, because neutral substitutions do not change the distances, the difference $r_+-r_-$ represents the change of the distance (i.e. the change $d_H(a,c)- d_H(b,c)=n-m$, if $r\leq m$).  Thus,
\[
r_+-r_-=n-\max\{r,m\}
\]
These conditions are the required guards in the summation~(\ref{eq:h-intersection-a}).\qed
\end{proof}

\section{Memoryless communication}
\label{sec:memoryless}

Let us consider an $X\times Y$-valued stochastic process $\{(x_t,y_t)\}_{t\geq0}$, where $(X,\mathcal{X})$ and $(Y,\mathcal{Y})$ are measurable sets.  Our interest is in the `similarity' between the marginal processes $\{x_t\}_{t\geq0}$ and $\{y_t\}_{t\geq0}$ under special assumptions on the communication between $X$ and $Y$.  Recall that a Markov \emph{transition kernel} from $(X,\mathcal{X})$ to $(Y,\mathcal{Y})$ is a conditional probability measure $P(Y_i\mid x)$ on $(Y,\mathcal{Y})$, which is $\mathcal{X}$-measurable for each $Y_i\in\mathcal{Y}$.  We shall use measure-theoretic notation $dP(y\mid x)$ for transition kernel $P(Y_i\mid x)=\int_{Y_i}dP(y\mid x)$.

\begin{proposition}
Let $(X,\mathcal{X})$ and $(Y,\mathcal{Y})$ be measurable sets, and let $\{(x_t,y_t)\}_{t\geq0}$ be a $X\times Y$-valued stochastic process such that elements of the marginal process $\{y_t\}_{t\geq0}$ are conditionally independent given $\{x_t\}_{t\geq0}$:
\[
dP(y_t,\ldots,y_0\mid x_t,\ldots,x_0)=dP(y_t\mid x_t)\cdots dP(y_0\mid x_0)
\]
Then transition kernel $dP(y_{t+1}\mid y_t)$ can be expressed as a composition of transition kernels $dP(x_t\mid y_t)$, $dP(x_{t+1}\mid x_t)$ and $dP(y_{t+1}\mid x_{t+1})$ as follows:
\[
dP(y_{t+1}\mid y_t)=\int\limits_{x_{t+1}\in X}\int\limits_{x_t\in X} dP(y_{t+1}\mid x_{t+1})\,dP(x_{t+1}\mid x_t)\,dP(x_t\mid y_t)
\]
This transition kernel has the following properties:
\begin{enumerate}
\item If $X$ and $Y$ are statistically independent, then $y_{t+1}\in Y$ is independent of $y_t\in Y$: $dP(y_{t+1}\mid y_t)=dP(y_{t+1})$
\item If $dP(x\mid y)$ corresponds to a function $x=h(y)$, then
\[
dP(y_{t+1}\mid y_t)=\frac{dP(y_{t+1})}{P\{h^{-1}\circ h(y_{t+1})\}}\,dP(x_{t+1}=h(y_{t+1})\mid x_t=h(y_t))
\]
\item If $dP(y\mid x)$ corresponds to a function $y=g(x)$, then
\[
dP(y_{t+1}\mid y_t)=\frac{1}{P\{g^{-1}(y_t)\}}\int\limits_{x_{t+1}\in g^{-1}(y_{t+1})}\int\limits_{x_t\in g^{-1}(y_t)}dP(x_{t+1}\mid x_t)\,dP(x_t)
\]
\item If $dP(y\mid x)$ corresponds to a bijection $y=h(x)$, then
\[
dP(y_{t+1}\mid y_t)=dP(x_{t+1}=h(y_{t+1})\mid x_t=h(y_t))
\]
\end{enumerate}
\label{pr:induced-kernel}
\end{proposition}

\begin{proof}
Transition kernel $dP(x_{t+1}\mid x_t)$ can generally be expressed as follows:
\[
dP(y_{t+1}\mid y_t)
=\int\limits_{x_{t+1}\in X}\int\limits_{x_t\in X}dP(y_{t+1}\mid x_{t+1},x_t,y_t)\,dP(x_{t+1}\mid x_t,y_t)\,dP(x_t\mid y_t)
\]
Using the Bayes formula and conditional independence gives
\begin{eqnarray*}
dP(y_{t+1}\mid x_{t+1},x_t,y_t)&=&\frac{dP(y_{t+1},y_t\mid x_{t+1},x_t)}{\int\limits_{y_{t+1}\in Y}dP(y_{t+1},y_t\mid x_{t+1},x_t)}\\
&=&\frac{dP(y_{t+1}\mid x_{t+1})\,dP(y_t\mid x_t)}{\int\limits_{y_{t+1}\in Y}dP(y_{t+1}\mid x_{t+1})\,dP(y_t\mid x_t)}=dP(y_{t+1}\mid x_{t+1})
\end{eqnarray*}
\begin{eqnarray*}
dP(x_{t+1}\mid x_t,y_t)&=&\int\limits_{y_{t+1}\in Y}dP(y_{t+1},x_{t+1}\mid x_t,y_t)\\
&=&\int\limits_{y_{t+1}\in Y}\frac{dP(y_{t+1},y_t\mid x_{t+1},x_t)\,dP(x_{t+1}\mid x_t)}{dP(y_t\mid x_t)}\\
&=&\int\limits_{y_{t+1}\in Y}\frac{dP(y_{t+1}\mid x_{t+1})\,dP(y_t\mid x_t)\,dP(x_{t+1}\mid x_t)}{dP(y_t\mid x_t)}\\
&=&dP(x_{t+1}\mid x_t)
\end{eqnarray*}
Thus, $dP(y_{t+1}\mid y_t)$ can be expressed using the composition of transition kernels $dP(y_{t+1}\mid x_{t+1})\,dP(x_{t+1}\mid x_t)\,dP(x_t\mid y_t)$.  We now consider four important cases.
\begin{enumerate}
\item If $X$ and $Y$ are independent, then $dP(y_{t+1}\mid x_{t+1})=dP(y_{t+1})$ and $dP(x_t\mid y_t)=dP(x_t)$, and therefore
\[
dP(y_{t+1}\mid y_t)=dP(y_{t+1})\int\limits_{x_{t+1}\in X}\int\limits_{x_t\in X}dP(x_{t+1}\mid x_t)\,dP(x_t)=dP(y_{t+1})
\]
\item If $x=h(y)$, then
\[
dP(x_t\mid y_t)=\delta_{h(y_t)}(x_t)\,,\qquad
dP(y_{t+1}\mid x_{t+1})=\frac{dP(y_{t+1})}{P\{h^{-1}\circ h(y_{t+1})\}}
\]
which gives the resulting expression.
\item If $y=g(x)$, then
\[
dP(x_t\mid y_t)=\frac{dP(x_t)}{P\{g^{-1}(y_t)\}}\,,\qquad
dP(y_{t+1}\mid x_{t+1})=\delta_{g(x_{t+1})}(y_{t+1})
\]
The resulting expression is obtained by integrating $dP(x_{t+1}\mid x_t)$ for each $x_{t+1}\in g^{-1}(y_{t+1})$ and $x_t\in g^{-1}(y_t)$.
\item Follows from $h^{-1}\circ h(y)=y$ for a bijection.
\end{enumerate}\qed
\end{proof}

\begin{remark}
\label{rem:cond-independence}
The assumption of conditional independence $dP(y_t,\ldots,y_0\mid x_t,\ldots,x_0)=dP(y_t\mid x_t)\cdots dP(y_0\mid x_0)$ is common in theory of non-linear filtering or hidden Markov models \citep{Stratonovich59:_nonlinear}, and it is equivalent to the assumption that the observed process $\{y_t\}_{t\geq0}$ does not provide more information about the hidden process $\{x_t\}_{t\geq0}$ than $x$ itself:
\[
dP(x_{t-1}\mid x_t,y_t)=dP(x_{t-1}\mid x_t)
\]
The idea is that $y$ is a `noisy' version of $x$.  Then, using the Bayes formula gives:
\[
dP(y_t\mid x_t,x_{t-1})=\frac{dP(x_{t-1}\mid x_t,y_t)}{dP(x_{t-1}\mid x_t)}\,dP(y_t\mid x_t)=dP(y_t\mid x_t)
\]
Note that it is also very common to assume that the hidden process $\{x_t\}_{t\geq0}$ is Markov, and together the two assumptions imply that the joint process $\{(x_t,y_t)\}_{t\geq0}$ is also Markov (but the observed process $\{y_t\}_{t\geq0}$ is usually not Markov, but a conditional Markov process).  Note that it is not required in Proposition~\ref{pr:induced-kernel} for any of the stochastic processes $\{(x_t,y_t)\}_{t\geq0}$, $\{x_t\}_{t\geq0}$ or $\{y_t\}_{t\geq0}$ to be Markov.  In the context of Section~\ref{sec:monotonic}, the unobserved variable $x\in X$ is the distance to optimum $d(\top,\omega)$, and observed variable $y\in Y$ is fitness.
\end{remark}

\section{Proof of Theorem~\ref{th:weak-mon}}
\label{sec:weak-mon}

Let us consider a sequence $\{B_n\}_{n\in\bN}$ of closed balls $B_n:=B[\top,\delta_n]=\{\omega:d(\top,\omega)\leq\delta_n\}$ of decreasing radia $\{\delta_n\}_{n\in\bN}$ in a metric space $(\Omega,d)$, such that the balls are nested:
\[
B_1\supset B_2\supset\cdots\supset B_{n-1}\supset B_n\supset\cdots
\]
The difference $S_n:=B_{n-1}\setminus B_n$ will be referred to as a `sphere' of radius $\delta_n$.  Let $P$ be a Borel probability measure on $\Omega$ (i.e. related to the topology on $\Omega$).  We shall denote by $P_n$ the conditional probability $P(E\mid\omega\in S_n)$ associated with the sphere $S_n$:
\[
P_n\{E\}:=\int_{E\cap S_n}dP(\omega\mid\omega\in S_n)=\frac{P\{E\cap S_n\}}{P\{S_n\}}
\]
Given a $P$-measurable function $f:\Omega\to\bR$, we shall denote by $\bE[f(\omega_n)]$ the conditional expectation $\bE\{f(\omega)\mid\omega\in S_n\}$:
\[
\bE[f(\omega_n)]:=\int_{S_n} f(\omega)\,dP(\omega\mid\omega\in S_n)
\]
For example, if $P$ is a uniform distribution, then $P_n(\omega)=1/|S_n|$ and $\bE[f(\omega_n)]$ is the average value of $f(\omega)$ in $S_n$.

Similarly, we consider a sequence $\{Y_n\}_{n\in\bN}$ of intervals $Y_n\subseteq\bR$ and conditional probabilities $P(E\mid\omega\in f^{-1}(Y_n))$ defined by their pre-images $f^{-1}(Y_n)=\{\omega:f(\omega)\in Y_n\}$.  We shall denote by $\bE[d(\top,\omega_n)]$ the conditional expectation $\bE\{d(\top,\omega)\mid\omega\in f^{-1}(Y_n)\}$:
\[
\bE[d(\top,\omega_n)]:=\int_{f^{-1}(Y_n)} d(\top,\omega)\,dP(\omega\mid\omega\in f^{-1}(Y_n))
\]
For example, if $P$ is a uniform distribution, then $\bE[d(\top,\omega_n)]$ is the average distance $d(\top,\omega)$ in $f^{-1}(Y_n)$.  We now prove the theorem.

\begin{proof}[Theorem~\ref{th:weak-mon}]
For convenience, we shall assume that $\top$ is a global optimum, so that condition $\omega\in E\subseteq\Omega$ can be omitted.  If $\Omega$ contains multiple optima $\top$, then by $d(\top,\omega)$ we understand the distance to the nearest optimum.
\begin{description}
\item [($\Rightarrow$)] Let $\{\varepsilon_n\}_{n\in\bN}$ be a decreasing sequence of $\varepsilon_n>0$ (i.e. $\varepsilon_n>\varepsilon_{n+1}$), and let $\delta_n>0$ be defined as follows:
\[
\delta_n:=\sup\{\delta>0:d(\top,\omega)<\delta\ \Rightarrow\ f(\omega)>f(\top)-\varepsilon_n\}
\]
Such $\delta_n>0$ exists for each $\varepsilon_n>0$ by continuity of $f$ at $\top$, and because $f(\top)\geq f(\omega)$ for all $\omega$.  The sequence $\{\delta_n\}_{n\in\bN}$ is also decreasing (otherwise $\{\varepsilon_n\}_{n\in\bN}$ is non-decreasing or $f$ is not continuous at $\top$).  Observe that we already have monotonicity with respect to balls $B_n\subset B_{n-1}$ in the following sense:
\[
\delta_n\geq\delta_m\quad\Rightarrow\quad\varepsilon_n\geq\varepsilon_m
\]
We need to prove monotonicity of conditional expectations $\bE[f(\omega_n)]$ within spheres $S_n=B_{n-1}\setminus B_n$.  Let $n\leq m$ so that $d(\top,\omega_n)\geq d(\top,\omega_m)$.  There are two mutually excluding cases:
\[
\bE[f(\omega_n)]\leq\bE[f(\omega_m)]\quad\mbox{or}\quad
\bE[f(\omega_n)]>\bE[f(\omega_m)]
\]
The first case corresponds to monotonicity and non-negative difference $\bE[f(\omega_m)]-\bE[f(\omega_n)]\geq0$; otherwise, the difference is negative $\bE[f(\omega_m)]-\bE[f(\omega_n)]<0$.  Using the Markov inequality for conditional probability $P_n\{f(\top)-f(\omega)\geq\varepsilon_m\}$ and the fact that $\bE[f(\omega_n)]>f(\top)-\varepsilon_n$ we derive the following bounds for $\bE[f(\omega_n)]$:
\[
\varepsilon_mP_n\{f(\top)-f(\omega)\geq\varepsilon_m\}\leq f(\top)-\bE[f(\omega_n)]<\varepsilon_n
\]
On the other hand $0\leq f(\top)-\bE[f(\omega_m)]<\varepsilon_m$.  These inequalities allow us to give the following bounds on the difference $\bE[f(\omega_m)]-\bE[f(\omega_n)]$:
\[
\varepsilon_mP_n\{f(\top)-f(\omega)\geq\varepsilon_m\}-\varepsilon_m<\bE[f(\omega_m)]-\bE[f(\omega_n)]<\varepsilon_n
\]
Substituting $P_n\{f(\top)-f(\omega)\geq\varepsilon_m\}=1-P_n\{f(\top)-f(\omega)<\varepsilon_m\}$ we obtain
\[
-\varepsilon_mP_n\{f(\top)-f(\omega)<\varepsilon_m\}<\bE[f(\omega_m)]-\bE[f(\omega_n)]<\varepsilon_n
\]
Conditional probability $P_n\{f(\top)-f(\omega)<\varepsilon_m\}$ converges to zero for any decreasing sequence $\varepsilon_m\to0$, which proves that the probability of non-negative difference $\bE[f(\omega_m)]-\bE[f(\omega_n)]\geq0$ converges to one.  In other words, the implication $d(\top,a)\geq d(\top,b)\Rightarrow\bE[f(a)]\leq\bE[f(b)]$ is true with probability one as $d(\top,b)\to0$.

\item [($\Leftarrow$)] Consider the function $d_\top(\omega):=d(\top,\omega)=\delta$.  The pre-image $d_\top^{-1}([0,\delta))=\{\omega:d(\top,\omega)<\delta\}$ of each open interval $[0,\delta)$ is an open ball $B[\top,\delta)$.  Because $f$ maps open balls $B[\top,\delta)$ to open intervals $(f(\top)-\varepsilon,f(\top)]$ by our assumption, the composition $f\circ d_\top^{-1}$ is an open mapping of open intervals $[0,\delta)$.  Therefore, the inverse function $(f\circ d_\top^{-1})^{-1}=d_\top\circ f^{-1}$ is continuous at $f(\top)$.  This means that for any decreasing sequence $\{\delta_n\}_{n\in\bN}$ we can construct the corresponding decreasing sequence $\{\varepsilon_n\}_{n\in\bN}$ by setting
\[
\varepsilon_n:=\sup\{\varepsilon>0:f(\omega)>f(\top)-\varepsilon\ \Rightarrow\ d(\top,\omega)<\delta_n\}
\]
The rest of the proof is identical to that of the first implication.  Specifically, using Markov inequality we derive the following bounds for $\bE[d(\top,\omega_n)]$:
\[
\delta_mP_n\{d(\top,\omega)\geq\delta_m)\}\leq\bE[d(\top,\omega_n)]<\delta_n
\]
Using bounds $0\leq\bE[d(\top,\omega_m)]<\delta_m$ we obtain the following bounds on the difference $\bE[d(\top,\omega_n)]-\bE[d(\top,\omega_m)]$:
\[
-\delta_mP_n\{d(\top,\omega)<\delta_m\}\leq\bE[d(\top,\omega_n)]-\bE[d(\top,\omega_m)]<\delta_n
\]
Because conditional probability $P_n\{d(\top,\omega)<\delta_m\}$ converges to zero for any decreasing sequence $\delta_m\to0$, this proves that the probability of non-negative difference $\bE[d(\top,\omega_n)]-\bE[d(\top,\omega_m)]\geq0$ converges to one.  In other words, the implication $f(a)\leq f(b)\Rightarrow\bE[d(\top,a)]\geq\bE[d(\top,b)]$ is true with probability one as $f(b)\to\sup f$.

\item [($\iff$)] If $f$ is both continuous and open at $\top$, then both implications are true in probability, which means that $f$ and $d$ are weakly isomorphic at $\top$.
\end{description}\qed
\end{proof}

\section{Supplementary figures}
\label{sec:supp-figures}

\includepdf[pages={1-3}]{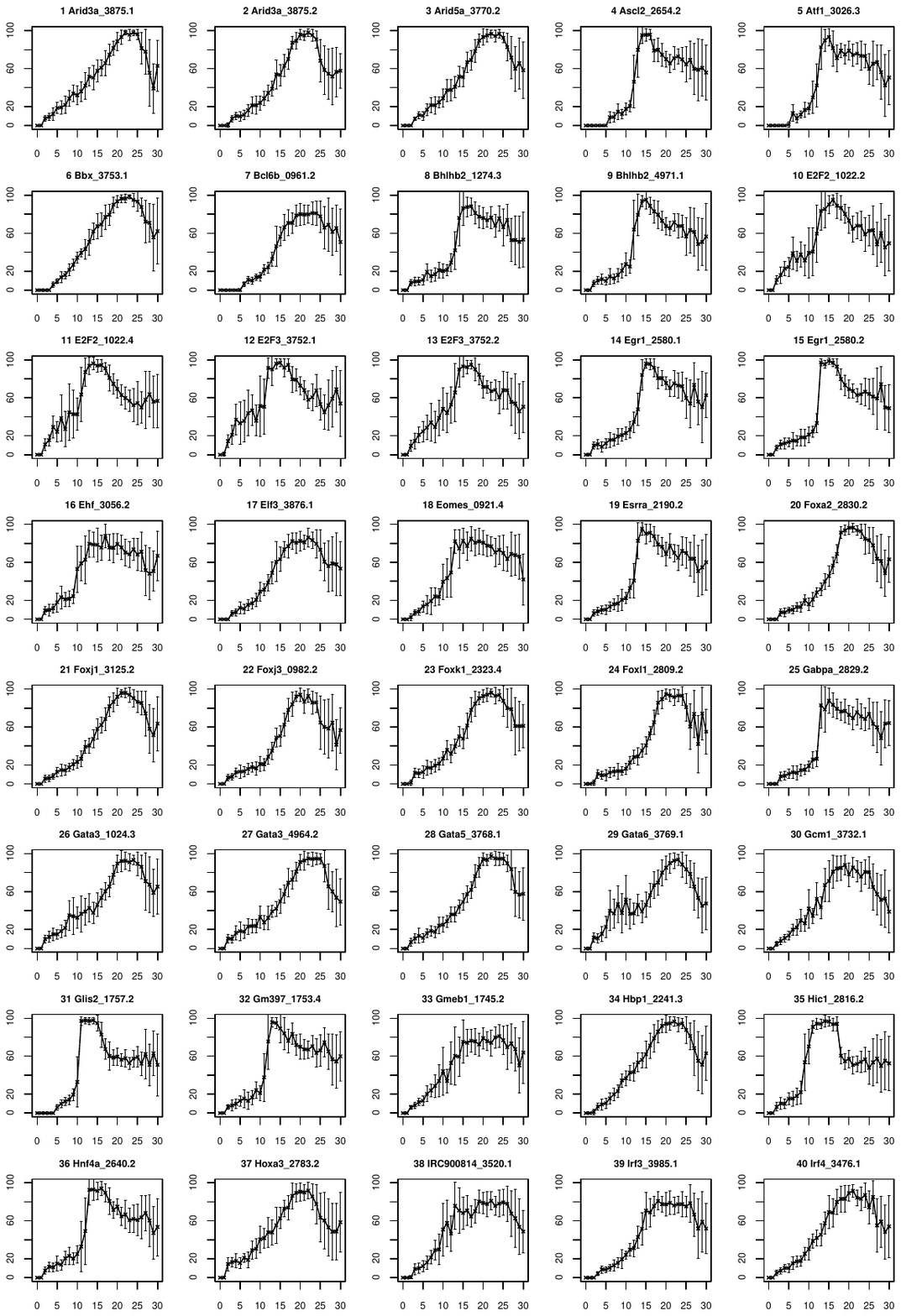}
\begin{figure}[!hb]
\caption{Optimal mutation rate control functions evolved by a Meta-GA on the transcription factor DNA-binding landscapes from \citep{Badis09}.  Ordinates show mutation rates, and abscissae show the binding scores.  Each panel corresponds to a different transcription factor.  Lines connect the average mutation rates obtained in 16 independent trials on a particular landscape.  Errorbars represent standard deviations from the mean.  The GAs do not spend much time at low binding scores meaning that the results become more random.}
\label{fig:all-curves}
\end{figure}

\includepdf[pages={1-3}]{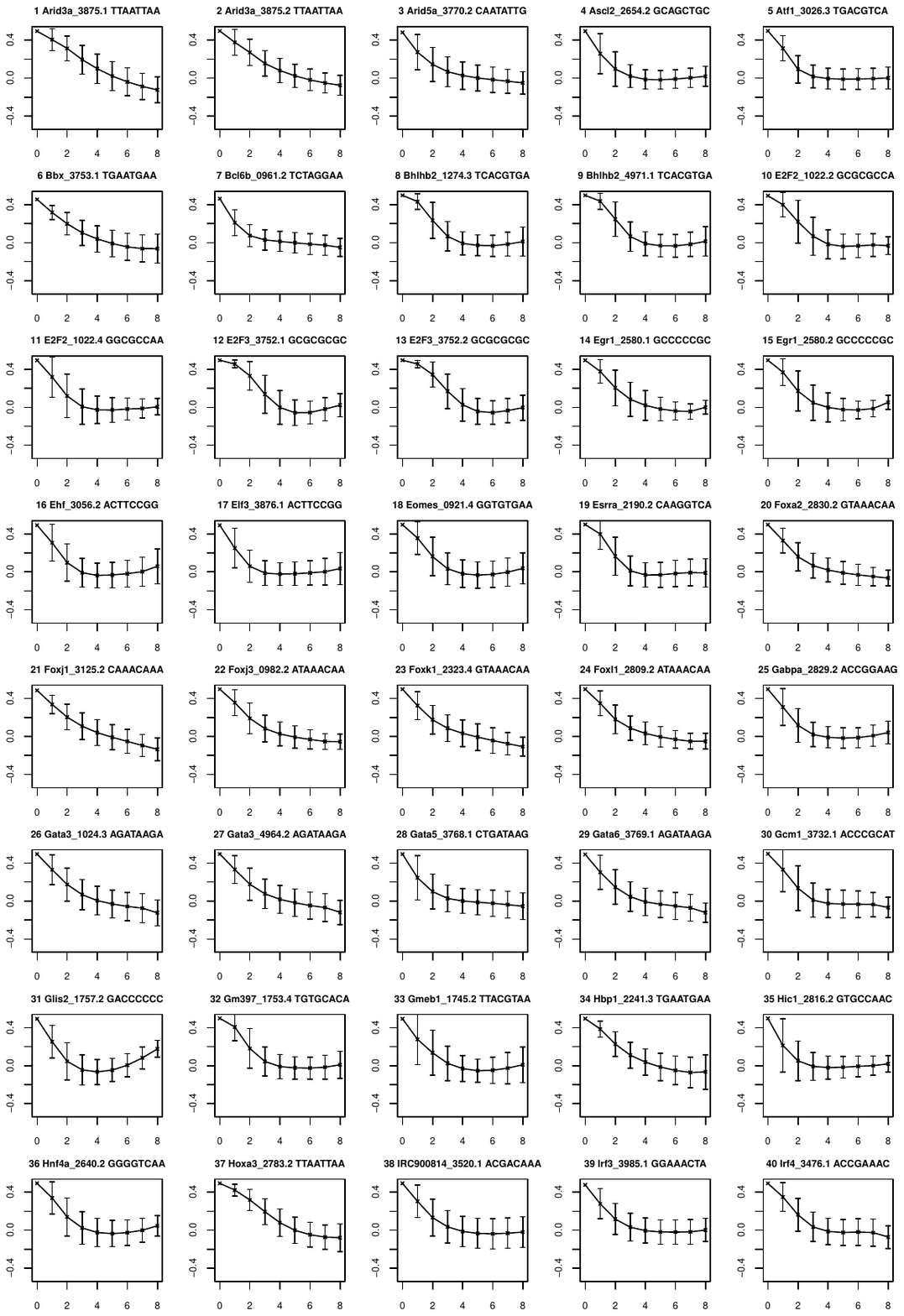}
\begin{figure}[!hb]
\caption{Landscapes of binding score between 8 base-pair DNA sequences and transcription factors (TF) from \citep{Badis09}.  Ordinates show binding scores, and abscissae show Hamming distances from the top string (a sequence with the highest DNA-TF binding score).  Each panel corresponds to a different transcription factor.  Lines connect mean values of the binding score for each value of the Hamming distance from the top string.  Errorbars represent standard deviations.  Note that this dataset does not distinguish between sequences on opposite strands of the DNA.  Therefore, a sequence and its reverse complement are shown only once and the Hamming distance shown is either to the top string or its reverse complement, whichever is the closer.}
\label{fig:all-fitness-local}
\end{figure}

\end{document}

%% file: fisher4-k-bw.tex
\begingroup
  \makeatletter
  \providecommand\color[2][]{%
    \GenericError{(gnuplot) \space\space\space\@spaces}{%
      Package color not loaded in conjunction with
      terminal option `colourtext'%
    }{See the gnuplot documentation for explanation.%
    }{Either use 'blacktext' in gnuplot or load the package
      color.sty in LaTeX.}%
    \renewcommand\color[2][]{}%
  }%
  \providecommand\includegraphics[2][]{%
    \GenericError{(gnuplot) \space\space\space\@spaces}{%
      Package graphicx or graphics not loaded%
    }{See the gnuplot documentation for explanation.%
    }{The gnuplot epslatex terminal needs graphicx.sty or graphics.sty.}%
    \renewcommand\includegraphics[2][]{}%
  }%
  \providecommand\rotatebox[2]{#2}%
  \@ifundefined{ifGPcolor}{%
    \newif\ifGPcolor
    \GPcolorfalse
  }{}%
  \@ifundefined{ifGPblacktext}{%
    \newif\ifGPblacktext
    \GPblacktexttrue
  }{}%
  \let\gplgaddtomacro\g@addto@macro
  \gdef\gplbacktext{}%
  \gdef\gplfronttext{}%
  \makeatother
  \ifGPblacktext
    \def\colorrgb#1{}%
    \def\colorgray#1{}%
  \else
    \ifGPcolor
      \def\colorrgb#1{\color[rgb]{#1}}%
      \def\colorgray#1{\color[gray]{#1}}%
      \expandafter\def\csname LTw\endcsname{\color{white}}%
      \expandafter\def\csname LTb\endcsname{\color{black}}%
      \expandafter\def\csname LTa\endcsname{\color{black}}%
      \expandafter\def\csname LT0\endcsname{\color[rgb]{1,0,0}}%
      \expandafter\def\csname LT1\endcsname{\color[rgb]{0,1,0}}%
      \expandafter\def\csname LT2\endcsname{\color[rgb]{0,0,1}}%
      \expandafter\def\csname LT3\endcsname{\color[rgb]{1,0,1}}%
      \expandafter\def\csname LT4\endcsname{\color[rgb]{0,1,1}}%
      \expandafter\def\csname LT5\endcsname{\color[rgb]{1,1,0}}%
      \expandafter\def\csname LT6\endcsname{\color[rgb]{0,0,0}}%
      \expandafter\def\csname LT7\endcsname{\color[rgb]{1,0.3,0}}%
      \expandafter\def\csname LT8\endcsname{\color[rgb]{0.5,0.5,0.5}}%
    \else
      \def\colorrgb#1{\color{black}}%
      \def\colorgray#1{\color[gray]{#1}}%
      \expandafter\def\csname LTw\endcsname{\color{white}}%
      \expandafter\def\csname LTb\endcsname{\color{black}}%
      \expandafter\def\csname LTa\endcsname{\color{black}}%
      \expandafter\def\csname LT0\endcsname{\color{black}}%
      \expandafter\def\csname LT1\endcsname{\color{black}}%
      \expandafter\def\csname LT2\endcsname{\color{black}}%
      \expandafter\def\csname LT3\endcsname{\color{black}}%
      \expandafter\def\csname LT4\endcsname{\color{black}}%
      \expandafter\def\csname LT5\endcsname{\color{black}}%
      \expandafter\def\csname LT6\endcsname{\color{black}}%
      \expandafter\def\csname LT7\endcsname{\color{black}}%
      \expandafter\def\csname LT8\endcsname{\color{black}}%
    \fi
  \fi
  \setlength{\unitlength}{0.0500bp}%
  \begin{picture}(5760.00,4032.00)%
    \gplgaddtomacro\gplbacktext{%
      \csname LTb\endcsname%
      \put(814,704){\makebox(0,0)[r]{\strut{}$0$}}%
      \put(814,1317){\makebox(0,0)[r]{\strut{}$0.2$}}%
      \put(814,1929){\makebox(0,0)[r]{\strut{}$0.4$}}%
      \put(814,2542){\makebox(0,0)[r]{\strut{}$0.6$}}%
      \put(814,3154){\makebox(0,0)[r]{\strut{}$0.8$}}%
      \put(814,3767){\makebox(0,0)[r]{\strut{}$1$}}%
      \put(946,484){\makebox(0,0){\strut{}$0$}}%
      \put(1829,484){\makebox(0,0){\strut{}$20$}}%
      \put(2713,484){\makebox(0,0){\strut{}$40$}}%
      \put(3596,484){\makebox(0,0){\strut{}$60$}}%
      \put(4480,484){\makebox(0,0){\strut{}$80$}}%
      \put(5363,484){\makebox(0,0){\strut{}$100$}}%
      \put(176,2235){\rotatebox{-270}{\makebox(0,0){\strut{}Probability of adaptation}}}%
      \put(3154,154){\makebox(0,0){\strut{}Mutation radius $r$}}%
    }%
    \gplgaddtomacro\gplfronttext{%
      \csname LTb\endcsname%
      \put(4376,1757){\makebox(0,0)[r]{\strut{}$n = 90$}}%
      \csname LTb\endcsname%
      \put(4376,1537){\makebox(0,0)[r]{\strut{}$n = 80$}}%
      \csname LTb\endcsname%
      \put(4376,1317){\makebox(0,0)[r]{\strut{}$n = 75$}}%
      \csname LTb\endcsname%
      \put(4376,1097){\makebox(0,0)[r]{\strut{}$n = 50$}}%
      \csname LTb\endcsname%
      \put(4376,877){\makebox(0,0)[r]{\strut{}$n = 25$}}%
    }%
    \gplbacktext
    \put(0,0){\includegraphics{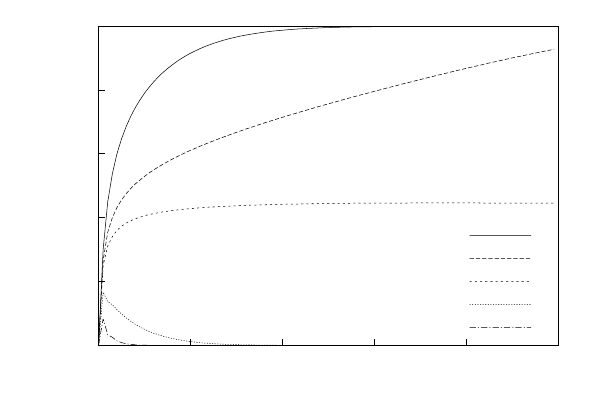}}%
    \gplfronttext
  \end{picture}%
\endgroup

%% file: mutation-rates-10-4-4-bw.tex
\begingroup
  \makeatletter
  \providecommand\color[2][]{%
    \GenericError{(gnuplot) \space\space\space\@spaces}{%
      Package color not loaded in conjunction with
      terminal option `colourtext'%
    }{See the gnuplot documentation for explanation.%
    }{Either use 'blacktext' in gnuplot or load the package
      color.sty in LaTeX.}%
    \renewcommand\color[2][]{}%
  }%
  \providecommand\includegraphics[2][]{%
    \GenericError{(gnuplot) \space\space\space\@spaces}{%
      Package graphicx or graphics not loaded%
    }{See the gnuplot documentation for explanation.%
    }{The gnuplot epslatex terminal needs graphicx.sty or graphics.sty.}%
    \renewcommand\includegraphics[2][]{}%
  }%
  \providecommand\rotatebox[2]{#2}%
  \@ifundefined{ifGPcolor}{%
    \newif\ifGPcolor
    \GPcolorfalse
  }{}%
  \@ifundefined{ifGPblacktext}{%
    \newif\ifGPblacktext
    \GPblacktexttrue
  }{}%
  \let\gplgaddtomacro\g@addto@macro
  \gdef\gplbacktext{}%
  \gdef\gplfronttext{}%
  \makeatother
  \ifGPblacktext
    \def\colorrgb#1{}%
    \def\colorgray#1{}%
  \else
    \ifGPcolor
      \def\colorrgb#1{\color[rgb]{#1}}%
      \def\colorgray#1{\color[gray]{#1}}%
      \expandafter\def\csname LTw\endcsname{\color{white}}%
      \expandafter\def\csname LTb\endcsname{\color{black}}%
      \expandafter\def\csname LTa\endcsname{\color{black}}%
      \expandafter\def\csname LT0\endcsname{\color[rgb]{1,0,0}}%
      \expandafter\def\csname LT1\endcsname{\color[rgb]{0,1,0}}%
      \expandafter\def\csname LT2\endcsname{\color[rgb]{0,0,1}}%
      \expandafter\def\csname LT3\endcsname{\color[rgb]{1,0,1}}%
      \expandafter\def\csname LT4\endcsname{\color[rgb]{0,1,1}}%
      \expandafter\def\csname LT5\endcsname{\color[rgb]{1,1,0}}%
      \expandafter\def\csname LT6\endcsname{\color[rgb]{0,0,0}}%
      \expandafter\def\csname LT7\endcsname{\color[rgb]{1,0.3,0}}%
      \expandafter\def\csname LT8\endcsname{\color[rgb]{0.5,0.5,0.5}}%
    \else
      \def\colorrgb#1{\color{black}}%
      \def\colorgray#1{\color[gray]{#1}}%
      \expandafter\def\csname LTw\endcsname{\color{white}}%
      \expandafter\def\csname LTb\endcsname{\color{black}}%
      \expandafter\def\csname LTa\endcsname{\color{black}}%
      \expandafter\def\csname LT0\endcsname{\color{black}}%
      \expandafter\def\csname LT1\endcsname{\color{black}}%
      \expandafter\def\csname LT2\endcsname{\color{black}}%
      \expandafter\def\csname LT3\endcsname{\color{black}}%
      \expandafter\def\csname LT4\endcsname{\color{black}}%
      \expandafter\def\csname LT5\endcsname{\color{black}}%
      \expandafter\def\csname LT6\endcsname{\color{black}}%
      \expandafter\def\csname LT7\endcsname{\color{black}}%
      \expandafter\def\csname LT8\endcsname{\color{black}}%
    \fi
  \fi
  \setlength{\unitlength}{0.0500bp}%
  \begin{picture}(5760.00,4032.00)%
    \gplgaddtomacro\gplbacktext{%
      \csname LTb\endcsname%
      \put(946,704){\makebox(0,0)[r]{\strut{}$0$}}%
      \put(946,1470){\makebox(0,0)[r]{\strut{}$0.25$}}%
      \put(946,2236){\makebox(0,0)[r]{\strut{}$0.5$}}%
      \put(946,3001){\makebox(0,0)[r]{\strut{}$0.75$}}%
      \put(946,3767){\makebox(0,0)[r]{\strut{}$1$}}%
      \put(1078,484){\makebox(0,0){\strut{}$0$}}%
      \put(1935,484){\makebox(0,0){\strut{}$2$}}%
      \put(2792,484){\makebox(0,0){\strut{}$4$}}%
      \put(3649,484){\makebox(0,0){\strut{}$6$}}%
      \put(4506,484){\makebox(0,0){\strut{}$8$}}%
      \put(5363,484){\makebox(0,0){\strut{}$10$}}%
      \put(176,2235){\rotatebox{-270}{\makebox(0,0){\strut{}Mutation rate $\mu$}}}%
      \put(3220,154){\makebox(0,0){\strut{}Distance to optimum $n=d(\top,a)$}}%
    }%
    \gplgaddtomacro\gplfronttext{%
      \csname LTb\endcsname%
      \put(3223,3657){\makebox(0,0)[r]{\strut{}Step}}%
      \csname LTb\endcsname%
      \put(3223,3437){\makebox(0,0)[r]{\strut{}Linear $n/l$}}%
      \csname LTb\endcsname%
      \put(3223,3217){\makebox(0,0)[r]{\strut{}$\max_\mu P_\mu(m<n\mid n)$}}%
      \csname LTb\endcsname%
      \put(3223,2997){\makebox(0,0)[r]{\strut{}$P_0(m<n)$}}%
    }%
    \gplbacktext
    \put(0,0){\includegraphics{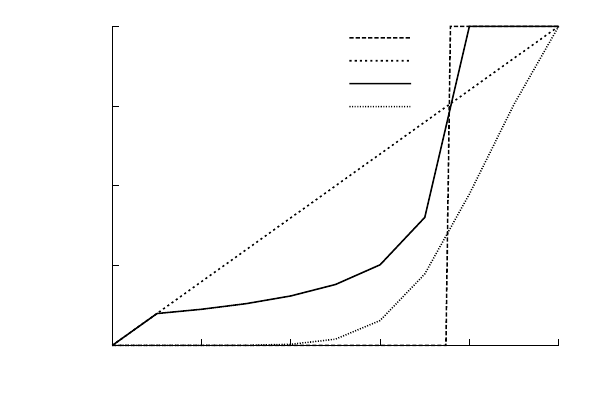}}%
    \gplfronttext
  \end{picture}%
\endgroup

%% file: mutation-rates-evolved-10-4-bw.tex
\begingroup
  \makeatletter
  \providecommand\color[2][]{%
    \GenericError{(gnuplot) \space\space\space\@spaces}{%
      Package color not loaded in conjunction with
      terminal option `colourtext'%
    }{See the gnuplot documentation for explanation.%
    }{Either use 'blacktext' in gnuplot or load the package
      color.sty in LaTeX.}%
    \renewcommand\color[2][]{}%
  }%
  \providecommand\includegraphics[2][]{%
    \GenericError{(gnuplot) \space\space\space\@spaces}{%
      Package graphicx or graphics not loaded%
    }{See the gnuplot documentation for explanation.%
    }{The gnuplot epslatex terminal needs graphicx.sty or graphics.sty.}%
    \renewcommand\includegraphics[2][]{}%
  }%
  \providecommand\rotatebox[2]{#2}%
  \@ifundefined{ifGPcolor}{%
    \newif\ifGPcolor
    \GPcolorfalse
  }{}%
  \@ifundefined{ifGPblacktext}{%
    \newif\ifGPblacktext
    \GPblacktexttrue
  }{}%
  \let\gplgaddtomacro\g@addto@macro
  \gdef\gplbacktext{}%
  \gdef\gplfronttext{}%
  \makeatother
  \ifGPblacktext
    \def\colorrgb#1{}%
    \def\colorgray#1{}%
  \else
    \ifGPcolor
      \def\colorrgb#1{\color[rgb]{#1}}%
      \def\colorgray#1{\color[gray]{#1}}%
      \expandafter\def\csname LTw\endcsname{\color{white}}%
      \expandafter\def\csname LTb\endcsname{\color{black}}%
      \expandafter\def\csname LTa\endcsname{\color{black}}%
      \expandafter\def\csname LT0\endcsname{\color[rgb]{1,0,0}}%
      \expandafter\def\csname LT1\endcsname{\color[rgb]{0,1,0}}%
      \expandafter\def\csname LT2\endcsname{\color[rgb]{0,0,1}}%
      \expandafter\def\csname LT3\endcsname{\color[rgb]{1,0,1}}%
      \expandafter\def\csname LT4\endcsname{\color[rgb]{0,1,1}}%
      \expandafter\def\csname LT5\endcsname{\color[rgb]{1,1,0}}%
      \expandafter\def\csname LT6\endcsname{\color[rgb]{0,0,0}}%
      \expandafter\def\csname LT7\endcsname{\color[rgb]{1,0.3,0}}%
      \expandafter\def\csname LT8\endcsname{\color[rgb]{0.5,0.5,0.5}}%
    \else
      \def\colorrgb#1{\color{black}}%
      \def\colorgray#1{\color[gray]{#1}}%
      \expandafter\def\csname LTw\endcsname{\color{white}}%
      \expandafter\def\csname LTb\endcsname{\color{black}}%
      \expandafter\def\csname LTa\endcsname{\color{black}}%
      \expandafter\def\csname LT0\endcsname{\color{black}}%
      \expandafter\def\csname LT1\endcsname{\color{black}}%
      \expandafter\def\csname LT2\endcsname{\color{black}}%
      \expandafter\def\csname LT3\endcsname{\color{black}}%
      \expandafter\def\csname LT4\endcsname{\color{black}}%
      \expandafter\def\csname LT5\endcsname{\color{black}}%
      \expandafter\def\csname LT6\endcsname{\color{black}}%
      \expandafter\def\csname LT7\endcsname{\color{black}}%
      \expandafter\def\csname LT8\endcsname{\color{black}}%
    \fi
  \fi
    \setlength{\unitlength}{0.0500bp}%
    \ifx\gptboxheight\undefined%
      \newlength{\gptboxheight}%
      \newlength{\gptboxwidth}%
      \newsavebox{\gptboxtext}%
    \fi%
    \setlength{\fboxrule}{0.5pt}%
    \setlength{\fboxsep}{1pt}%
\begin{picture}(5760.00,3456.00)%
    \gplgaddtomacro\gplbacktext{%
      \put(747,502){\makebox(0,0)[r]{\strut{}$0$}}%
      \put(747,1189){\makebox(0,0)[r]{\strut{}$0.25$}}%
      \put(747,1877){\makebox(0,0)[r]{\strut{}$0.5$}}%
      \put(747,2564){\makebox(0,0)[r]{\strut{}$0.75$}}%
      \put(747,3251){\makebox(0,0)[r]{\strut{}$1$}}%
      \put(849,316){\makebox(0,0){\strut{}$0$}}%
      \put(1770,316){\makebox(0,0){\strut{}$2$}}%
      \put(2691,316){\makebox(0,0){\strut{}$4$}}%
      \put(3611,316){\makebox(0,0){\strut{}$6$}}%
      \put(4532,316){\makebox(0,0){\strut{}$8$}}%
      \put(5453,316){\makebox(0,0){\strut{}$10$}}%
    }%
    \gplgaddtomacro\gplfronttext{%
      \put(144,1876){\rotatebox{-270}{\makebox(0,0){\strut{}Mutation rate $\mu$}}}%
      \put(3151,130){\makebox(0,0){\strut{}Distance to optimum $n=d(\top,\omega)$}}%
    }%
    \gplbacktext
    \put(0,0){\includegraphics{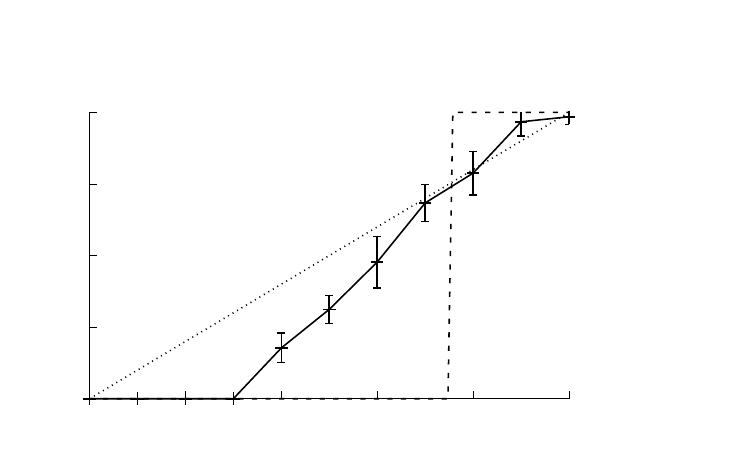}}%
    \gplfronttext
  \end{picture}%
\endgroup

%% file: landscape-types-mutation-bw.tex
\begingroup
  \makeatletter
  \providecommand\color[2][]{%
    \GenericError{(gnuplot) \space\space\space\@spaces}{%
      Package color not loaded in conjunction with
      terminal option `colourtext'%
    }{See the gnuplot documentation for explanation.%
    }{Either use 'blacktext' in gnuplot or load the package
      color.sty in LaTeX.}%
    \renewcommand\color[2][]{}%
  }%
  \providecommand\includegraphics[2][]{%
    \GenericError{(gnuplot) \space\space\space\@spaces}{%
      Package graphicx or graphics not loaded%
    }{See the gnuplot documentation for explanation.%
    }{The gnuplot epslatex terminal needs graphicx.sty or graphics.sty.}%
    \renewcommand\includegraphics[2][]{}%
  }%
  \providecommand\rotatebox[2]{#2}%
  \@ifundefined{ifGPcolor}{%
    \newif\ifGPcolor
    \GPcolorfalse
  }{}%
  \@ifundefined{ifGPblacktext}{%
    \newif\ifGPblacktext
    \GPblacktexttrue
  }{}%
  \let\gplgaddtomacro\g@addto@macro
  \gdef\gplbacktext{}%
  \gdef\gplfronttext{}%
  \makeatother
  \ifGPblacktext
    \def\colorrgb#1{}%
    \def\colorgray#1{}%
  \else
    \ifGPcolor
      \def\colorrgb#1{\color[rgb]{#1}}%
      \def\colorgray#1{\color[gray]{#1}}%
      \expandafter\def\csname LTw\endcsname{\color{white}}%
      \expandafter\def\csname LTb\endcsname{\color{black}}%
      \expandafter\def\csname LTa\endcsname{\color{black}}%
      \expandafter\def\csname LT0\endcsname{\color[rgb]{1,0,0}}%
      \expandafter\def\csname LT1\endcsname{\color[rgb]{0,1,0}}%
      \expandafter\def\csname LT2\endcsname{\color[rgb]{0,0,1}}%
      \expandafter\def\csname LT3\endcsname{\color[rgb]{1,0,1}}%
      \expandafter\def\csname LT4\endcsname{\color[rgb]{0,1,1}}%
      \expandafter\def\csname LT5\endcsname{\color[rgb]{1,1,0}}%
      \expandafter\def\csname LT6\endcsname{\color[rgb]{0,0,0}}%
      \expandafter\def\csname LT7\endcsname{\color[rgb]{1,0.3,0}}%
      \expandafter\def\csname LT8\endcsname{\color[rgb]{0.5,0.5,0.5}}%
    \else
      \def\colorrgb#1{\color{black}}%
      \def\colorgray#1{\color[gray]{#1}}%
      \expandafter\def\csname LTw\endcsname{\color{white}}%
      \expandafter\def\csname LTb\endcsname{\color{black}}%
      \expandafter\def\csname LTa\endcsname{\color{black}}%
      \expandafter\def\csname LT0\endcsname{\color{black}}%
      \expandafter\def\csname LT1\endcsname{\color{black}}%
      \expandafter\def\csname LT2\endcsname{\color{black}}%
      \expandafter\def\csname LT3\endcsname{\color{black}}%
      \expandafter\def\csname LT4\endcsname{\color{black}}%
      \expandafter\def\csname LT5\endcsname{\color{black}}%
      \expandafter\def\csname LT6\endcsname{\color{black}}%
      \expandafter\def\csname LT7\endcsname{\color{black}}%
      \expandafter\def\csname LT8\endcsname{\color{black}}%
    \fi
  \fi
    \setlength{\unitlength}{0.0500bp}%
    \ifx\gptboxheight\undefined%
      \newlength{\gptboxheight}%
      \newlength{\gptboxwidth}%
      \newsavebox{\gptboxtext}%
    \fi%
    \setlength{\fboxrule}{0.5pt}%
    \setlength{\fboxsep}{1pt}%
\begin{picture}(6840.00,4104.00)%
    \gplgaddtomacro\gplbacktext{%
      \csname LTb\endcsname%
      \put(645,557){\makebox(0,0)[r]{\strut{}$0$}}%
      \put(645,1162){\makebox(0,0)[r]{\strut{}$0.25$}}%
      \put(645,1767){\makebox(0,0)[r]{\strut{}$0.5$}}%
      \put(645,2373){\makebox(0,0)[r]{\strut{}$0.75$}}%
      \put(645,2978){\makebox(0,0)[r]{\strut{}$1$}}%
      \put(6533,371){\makebox(0,0){\strut{}$-0.4$}}%
      \put(5890,371){\makebox(0,0){\strut{}$-0.3$}}%
      \put(5247,371){\makebox(0,0){\strut{}$-0.2$}}%
      \put(4604,371){\makebox(0,0){\strut{}$-0.1$}}%
      \put(3961,371){\makebox(0,0){\strut{}$0$}}%
      \put(3319,371){\makebox(0,0){\strut{}$0.1$}}%
      \put(2676,371){\makebox(0,0){\strut{}$0.2$}}%
      \put(2033,371){\makebox(0,0){\strut{}$0.3$}}%
      \put(1390,371){\makebox(0,0){\strut{}$0.4$}}%
      \put(747,371){\makebox(0,0){\strut{}$0.5$}}%
      \put(3061,3769){\makebox(0,0){\strut{}$0.36$}}%
      \put(3383,3769){\makebox(0,0){\strut{}$0.41$}}%
      \put(5376,3769){\makebox(0,0){\strut{}$0.72$}}%
    }%
    \gplgaddtomacro\gplfronttext{%
      \csname LTb\endcsname%
      \put(144,2070){\rotatebox{-270}{\makebox(0,0){\strut{}Mutation rate $\mu$}}}%
      \put(3640,130){\makebox(0,0){\strut{}DNA Binding Score}}%
      \put(3640,4009){\makebox(0,0){\strut{}Monotonicity radius $\varepsilon$}}%
      \csname LTb\endcsname%
      \put(1461,2256){\makebox(0,0)[r]{\strut{}Srf}}%
      \csname LTb\endcsname%
      \put(1461,2070){\makebox(0,0)[r]{\strut{}Glis2}}%
      \csname LTb\endcsname%
      \put(1461,1884){\makebox(0,0)[r]{\strut{}Zfp740}}%
    }%
    \gplbacktext
    \put(0,0){\includegraphics{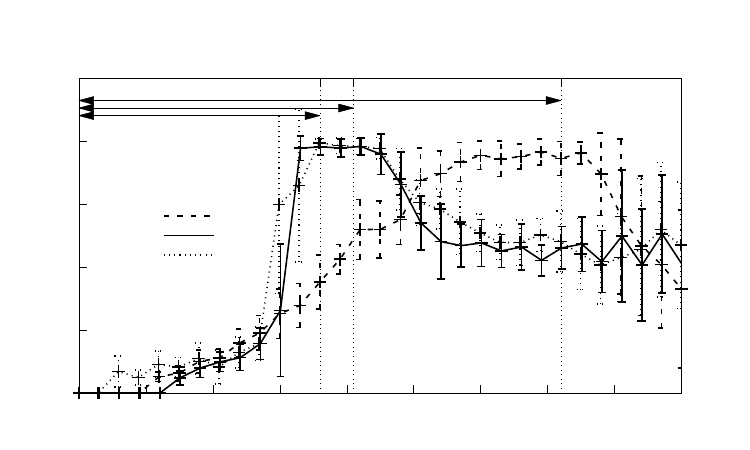}}%
    \gplfronttext
  \end{picture}%
\endgroup

%% file: landscape-types-fitness-bw.tex
\begingroup
  \makeatletter
  \providecommand\color[2][]{%
    \GenericError{(gnuplot) \space\space\space\@spaces}{%
      Package color not loaded in conjunction with
      terminal option `colourtext'%
    }{See the gnuplot documentation for explanation.%
    }{Either use 'blacktext' in gnuplot or load the package
      color.sty in LaTeX.}%
    \renewcommand\color[2][]{}%
  }%
  \providecommand\includegraphics[2][]{%
    \GenericError{(gnuplot) \space\space\space\@spaces}{%
      Package graphicx or graphics not loaded%
    }{See the gnuplot documentation for explanation.%
    }{The gnuplot epslatex terminal needs graphicx.sty or graphics.sty.}%
    \renewcommand\includegraphics[2][]{}%
  }%
  \providecommand\rotatebox[2]{#2}%
  \@ifundefined{ifGPcolor}{%
    \newif\ifGPcolor
    \GPcolorfalse
  }{}%
  \@ifundefined{ifGPblacktext}{%
    \newif\ifGPblacktext
    \GPblacktexttrue
  }{}%
  \let\gplgaddtomacro\g@addto@macro
  \gdef\gplbacktext{}%
  \gdef\gplfronttext{}%
  \makeatother
  \ifGPblacktext
    \def\colorrgb#1{}%
    \def\colorgray#1{}%
  \else
    \ifGPcolor
      \def\colorrgb#1{\color[rgb]{#1}}%
      \def\colorgray#1{\color[gray]{#1}}%
      \expandafter\def\csname LTw\endcsname{\color{white}}%
      \expandafter\def\csname LTb\endcsname{\color{black}}%
      \expandafter\def\csname LTa\endcsname{\color{black}}%
      \expandafter\def\csname LT0\endcsname{\color[rgb]{1,0,0}}%
      \expandafter\def\csname LT1\endcsname{\color[rgb]{0,1,0}}%
      \expandafter\def\csname LT2\endcsname{\color[rgb]{0,0,1}}%
      \expandafter\def\csname LT3\endcsname{\color[rgb]{1,0,1}}%
      \expandafter\def\csname LT4\endcsname{\color[rgb]{0,1,1}}%
      \expandafter\def\csname LT5\endcsname{\color[rgb]{1,1,0}}%
      \expandafter\def\csname LT6\endcsname{\color[rgb]{0,0,0}}%
      \expandafter\def\csname LT7\endcsname{\color[rgb]{1,0.3,0}}%
      \expandafter\def\csname LT8\endcsname{\color[rgb]{0.5,0.5,0.5}}%
    \else
      \def\colorrgb#1{\color{black}}%
      \def\colorgray#1{\color[gray]{#1}}%
      \expandafter\def\csname LTw\endcsname{\color{white}}%
      \expandafter\def\csname LTb\endcsname{\color{black}}%
      \expandafter\def\csname LTa\endcsname{\color{black}}%
      \expandafter\def\csname LT0\endcsname{\color{black}}%
      \expandafter\def\csname LT1\endcsname{\color{black}}%
      \expandafter\def\csname LT2\endcsname{\color{black}}%
      \expandafter\def\csname LT3\endcsname{\color{black}}%
      \expandafter\def\csname LT4\endcsname{\color{black}}%
      \expandafter\def\csname LT5\endcsname{\color{black}}%
      \expandafter\def\csname LT6\endcsname{\color{black}}%
      \expandafter\def\csname LT7\endcsname{\color{black}}%
      \expandafter\def\csname LT8\endcsname{\color{black}}%
    \fi
  \fi
    \setlength{\unitlength}{0.0500bp}%
    \ifx\gptboxheight\undefined%
      \newlength{\gptboxheight}%
      \newlength{\gptboxwidth}%
      \newsavebox{\gptboxtext}%
    \fi%
    \setlength{\fboxrule}{0.5pt}%
    \setlength{\fboxsep}{1pt}%
\begin{picture}(6840.00,4104.00)%
    \gplgaddtomacro\gplbacktext{%
      \put(645,743){\makebox(0,0)[r]{\strut{}$-0.4$}}%
      \put(645,1374){\makebox(0,0)[r]{\strut{}$-0.2$}}%
      \put(645,2005){\makebox(0,0)[r]{\strut{}$0$}}%
      \put(645,2637){\makebox(0,0)[r]{\strut{}$0.2$}}%
      \put(645,3268){\makebox(0,0)[r]{\strut{}$0.4$}}%
      \put(645,3899){\makebox(0,0)[r]{\strut{}$0.6$}}%
      \put(747,557){\makebox(0,0){\strut{}$0$}}%
      \put(1428,557){\makebox(0,0){\strut{}$1$}}%
      \put(2108,557){\makebox(0,0){\strut{}$2$}}%
      \put(2789,557){\makebox(0,0){\strut{}$3$}}%
      \put(3470,557){\makebox(0,0){\strut{}$4$}}%
      \put(4151,557){\makebox(0,0){\strut{}$5$}}%
      \put(4831,557){\makebox(0,0){\strut{}$6$}}%
      \put(5512,557){\makebox(0,0){\strut{}$7$}}%
      \put(6193,557){\makebox(0,0){\strut{}$8$}}%
    }%
    \gplgaddtomacro\gplfronttext{%
      \put(144,2321){\rotatebox{-270}{\makebox(0,0){\strut{}DNA Binding Score}}}%
      \put(3640,316){\makebox(0,0){\strut{}Hamming distance to optimum $n=d_H(\top,\omega)$}}%
      \csname LTb\endcsname%
      \put(5745,3732){\makebox(0,0)[r]{\strut{}Srf}}%
      \csname LTb\endcsname%
      \put(5745,3546){\makebox(0,0)[r]{\strut{}Glis2}}%
      \csname LTb\endcsname%
      \put(5745,3360){\makebox(0,0)[r]{\strut{}Zfp740}}%
    }%
    \gplbacktext
    \put(0,0){\includegraphics{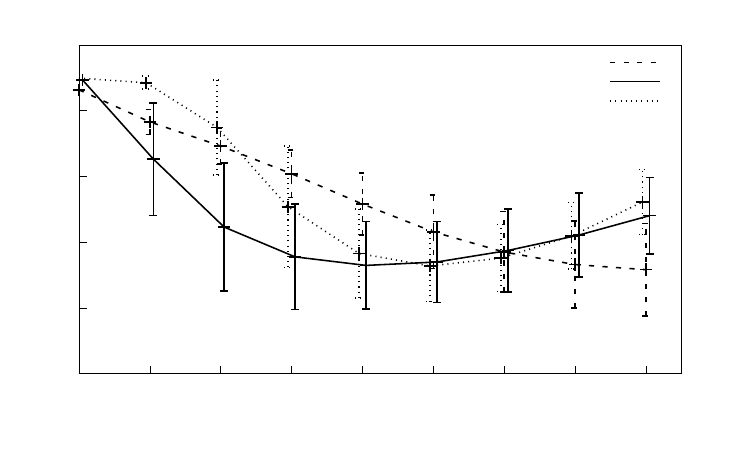}}%
    \gplfronttext
  \end{picture}%
\endgroup

%% file: edge-tau-rnd-bw.tex
\begingroup
  \makeatletter
  \providecommand\color[2][]{%
    \GenericError{(gnuplot) \space\space\space\@spaces}{%
      Package color not loaded in conjunction with
      terminal option `colourtext'%
    }{See the gnuplot documentation for explanation.%
    }{Either use 'blacktext' in gnuplot or load the package
      color.sty in LaTeX.}%
    \renewcommand\color[2][]{}%
  }%
  \providecommand\includegraphics[2][]{%
    \GenericError{(gnuplot) \space\space\space\@spaces}{%
      Package graphicx or graphics not loaded%
    }{See the gnuplot documentation for explanation.%
    }{The gnuplot epslatex terminal needs graphicx.sty or graphics.sty.}%
    \renewcommand\includegraphics[2][]{}%
  }%
  \providecommand\rotatebox[2]{#2}%
  \@ifundefined{ifGPcolor}{%
    \newif\ifGPcolor
    \GPcolorfalse
  }{}%
  \@ifundefined{ifGPblacktext}{%
    \newif\ifGPblacktext
    \GPblacktexttrue
  }{}%
  \let\gplgaddtomacro\g@addto@macro
  \gdef\gplbacktext{}%
  \gdef\gplfronttext{}%
  \makeatother
  \ifGPblacktext
    \def\colorrgb#1{}%
    \def\colorgray#1{}%
  \else
    \ifGPcolor
      \def\colorrgb#1{\color[rgb]{#1}}%
      \def\colorgray#1{\color[gray]{#1}}%
      \expandafter\def\csname LTw\endcsname{\color{white}}%
      \expandafter\def\csname LTb\endcsname{\color{black}}%
      \expandafter\def\csname LTa\endcsname{\color{black}}%
      \expandafter\def\csname LT0\endcsname{\color[rgb]{1,0,0}}%
      \expandafter\def\csname LT1\endcsname{\color[rgb]{0,1,0}}%
      \expandafter\def\csname LT2\endcsname{\color[rgb]{0,0,1}}%
      \expandafter\def\csname LT3\endcsname{\color[rgb]{1,0,1}}%
      \expandafter\def\csname LT4\endcsname{\color[rgb]{0,1,1}}%
      \expandafter\def\csname LT5\endcsname{\color[rgb]{1,1,0}}%
      \expandafter\def\csname LT6\endcsname{\color[rgb]{0,0,0}}%
      \expandafter\def\csname LT7\endcsname{\color[rgb]{1,0.3,0}}%
      \expandafter\def\csname LT8\endcsname{\color[rgb]{0.5,0.5,0.5}}%
    \else
      \def\colorrgb#1{\color{black}}%
      \def\colorgray#1{\color[gray]{#1}}%
      \expandafter\def\csname LTw\endcsname{\color{white}}%
      \expandafter\def\csname LTb\endcsname{\color{black}}%
      \expandafter\def\csname LTa\endcsname{\color{black}}%
      \expandafter\def\csname LT0\endcsname{\color{black}}%
      \expandafter\def\csname LT1\endcsname{\color{black}}%
      \expandafter\def\csname LT2\endcsname{\color{black}}%
      \expandafter\def\csname LT3\endcsname{\color{black}}%
      \expandafter\def\csname LT4\endcsname{\color{black}}%
      \expandafter\def\csname LT5\endcsname{\color{black}}%
      \expandafter\def\csname LT6\endcsname{\color{black}}%
      \expandafter\def\csname LT7\endcsname{\color{black}}%
      \expandafter\def\csname LT8\endcsname{\color{black}}%
    \fi
  \fi
  \setlength{\unitlength}{0.0500bp}%
  \begin{picture}(6480.00,6480.00)%
    \gplgaddtomacro\gplbacktext{%
      \csname LTb\endcsname%
      \put(946,957){\makebox(0,0)[r]{\strut{}$-0.4$}}%
      \put(946,1791){\makebox(0,0)[r]{\strut{}$-0.2$}}%
      \put(946,2625){\makebox(0,0)[r]{\strut{}$0$}}%
      \put(946,3459){\makebox(0,0)[r]{\strut{}$0.2$}}%
      \put(946,4294){\makebox(0,0)[r]{\strut{}$0.4$}}%
      \put(946,5128){\makebox(0,0)[r]{\strut{}$0.6$}}%
      \put(946,5962){\makebox(0,0)[r]{\strut{}$0.8$}}%
      \put(1078,737){\makebox(0,0){\strut{}$0.25$}}%
      \put(1533,737){\makebox(0,0){\strut{}$0.3$}}%
      \put(1988,737){\makebox(0,0){\strut{}$0.35$}}%
      \put(2443,737){\makebox(0,0){\strut{}$0.4$}}%
      \put(2898,737){\makebox(0,0){\strut{}$0.45$}}%
      \put(3353,737){\makebox(0,0){\strut{}$0.5$}}%
      \put(3808,737){\makebox(0,0){\strut{}$0.55$}}%
      \put(4263,737){\makebox(0,0){\strut{}$0.6$}}%
      \put(4718,737){\makebox(0,0){\strut{}$0.65$}}%
      \put(5173,737){\makebox(0,0){\strut{}$0.7$}}%
      \put(5628,737){\makebox(0,0){\strut{}$0.75$}}%
      \put(6083,737){\makebox(0,0){\strut{}$0.8$}}%
      \put(176,3459){\rotatebox{-270}{\makebox(0,0){\strut{}Landscape monotonicity $\tau$}}}%
      \put(3580,407){\makebox(0,0){\strut{}Mutation monotonicity radius $\varepsilon$}}%
      \put(5393,5335){\makebox(0,0){\strut{}Srf}}%
      \put(2567,1406){\makebox(0,0){\strut{}Glis2}}%
      \put(2066,1842){\makebox(0,0){\strut{}Zfp740}}%
    }%
    \gplgaddtomacro\gplfronttext{%
      \csname LTb\endcsname%
      \put(2043,5435){\makebox(0,0)[r]{\strut{}Data}}%
      \csname LTb\endcsname%
      \put(2043,5215){\makebox(0,0)[r]{\strut{}Linear fit}}%
    }%
    \gplbacktext
    \put(0,0){\includegraphics{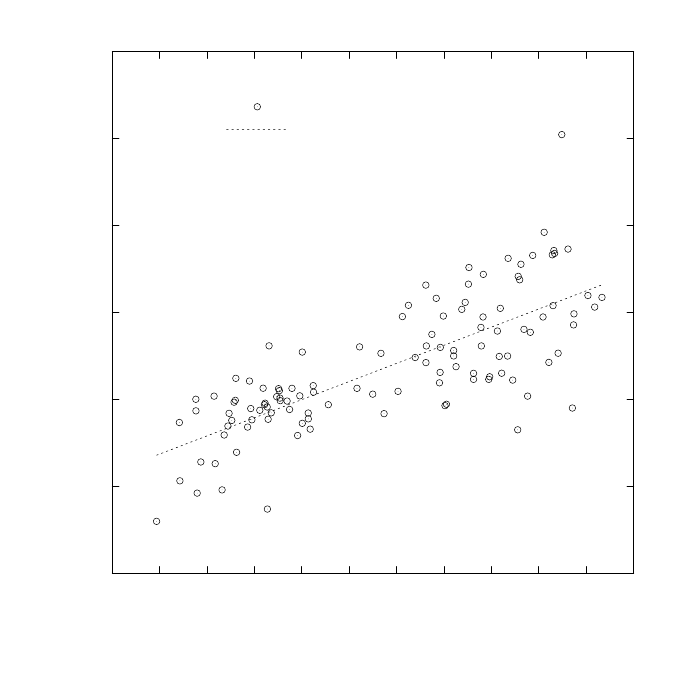}}%
    \gplfronttext
  \end{picture}%
\endgroup